\documentclass[11pt]{article}
\usepackage{graphicx}

\usepackage{lineno}

\title{Private Graph Property Testing}

\newif\ifshowauthors
\showauthorstrue
\ifshowauthors
\author{
Hendrik Fichtenberger\\
Google Research\\
fichtenberger@google.com
\and
Abigail Gentle\thanks{Research supported by an Australian government RTP scholarship.}\\
The University of Sydney\\
abigail.gentle@sydney.edu.au
\and
Tamalika Mukherjee\\
Max Planck Institute for Security and Privacy\\
tamalika.mukherjee@mpi-sp.org
\and
Sayantan Sen\thanks{Research supported by the NRF Investigatorship award (NRF-NRFI10-2024-0006)
and CQT Young Researcher Career Development Grant (25-YRCDG-SS).}\\
Centre for Quantum Technologies,\\
National University of Singapore\\
sayantan789@gmail.com
}
\else 
\author{
(Authors withheld)
}
\fi
\date{}

\usepackage{graphicx} %

\usepackage[utf8]{inputenc}
\usepackage{amsmath,amsthm,amssymb,xcolor}
\usepackage[margin=2.54cm]{geometry}
\usepackage{physics}
\usepackage[hidelinks]{hyperref}
\usepackage{float}
\usepackage{hyperref}
\usepackage{multirow}
\usepackage{enumitem}
\usepackage[capitalize]{cleveref}
\crefname{algocf}{algorithm}{algorithms} %
\Crefname{algocf}{Algorithm}{Algorithms} %
\usepackage{mathtools}
\usepackage{amssymb}
\usepackage{lineno}
\usepackage{dsfont}
\usepackage{booktabs}
\usepackage{comment}

\usepackage{pgfplots}
\usepackage[dvipsnames]{xcolor}
\pgfplotsset{compat=1.7}
\usetikzlibrary{tikzmark, arrows}
\usetikzlibrary{calc}
\usepackage{tkz-euclide}
\usepackage[normalem]{ulem}

\usepackage{amsthm}
\newtheorem{theorem}{Theorem}[section]
\newtheorem{definition}[theorem]{Definition}

\newtheorem{lemma}[theorem]{Lemma}
\newtheorem{remark}[theorem]{Remark}
\newtheorem{corollary}[theorem]{Corollary}
\usepackage[linesnumbered, boxed]{algorithm2e}
\usepackage{algpseudocode}
\newenvironment{proofof}[1]{\smallskip\noindent{\bf Proof of #1.}}%
{\hspace*{\fill}$\Box$\par}

\newcommand{\eps}{\varepsilon}
\newcommand{\poly}{\mathrm{poly}}

\usepackage{todonotes}
\setuptodonotes{inline}
\newcounter{mynotes}
  \newcommand{\priv}{\ensuremath{\varepsilon}\xspace}
  \newcommand{\pdelta}{\ensuremath{\delta}\xspace}
  \newcommand{\dst}{\ensuremath{\gamma}\xspace}
  \newcommand{\dist}{\dst}
  \newcommand{\failp}{\ensuremath{\beta}\xspace}
  \newcommand{\ns}{\ensuremath{n}\xspace}
  
  \newcommand{\pmech}{\cM\xspace}
  \newcommand{\defeq}{\coloneqq}
  
  \newcommand{\maxdeg}{\ensuremath{d}\xspace}
  \newcommand{\prop}{\ensuremath{\Pi}\xspace}
  \newcommand{\property}{\prop}
  \newcommand{\tester}{\ensuremath{\mathcal{T}}\xspace}
  \newcommand{\ms}{\ensuremath{\ell}\xspace}

  \newcommand{\eep}{{e^{\priv}}}
  \newcommand{\eepz}{{e^{\priv_{0}}}}
  \newcommand{\eepp}{{e^{\priv'}}}
  \newcommand{\ome}{\omega}
  \newcommand{\rd}{k} 
  \newcommand{\rad}{\operatorname{rd}}
  \newcommand{\sampler}{{\mathcal{S}}}
  \newcommand{\subsampler}{\sampler^{\text{sub}}}
  \newcommand{\cansampler}{\sampler^{\text{sub}}_{h}}
  \newcommand{\kdsampler}{\sampler^{\text{disc}}_{\ms, \rd}}
  \newcommand{\RR}{\ensuremath{\operatorname{RR}}\xspace}
  \newcommand{\blockdist}{\ensuremath{\lambda}\xspace}
  \newcommand{\diverge}{\ensuremath{\alpha}\xspace} %
  \newcommand{\indic}[1]{\mathds{1}\{ #1 \}}
  \newcommand{\transcript}{\ensuremath{W}\xspace}
  \newcommand{\spacetrans}{\ensuremath{\mathcal{W}}\xspace}
  \newcommand{\tsize}{\ensuremath{T}\xspace}
  \newcommand{\length}{L}
  \newcommand{\repeats}{r}
  
  \newcommand{\prox}{\phi}
  \newcommand{\bigO}[1]{O\!\left( #1 \right)}
  \newcommand{\bigOtilde}[1]{\widetilde{O}\!\left( #1 \right)}
  \newcommand{\bigOme}[1]{\Omega\!\left( #1 \right)}
  \newcommand{\bigOmetilde}[1]{\widetilde{\Omega}\!\left( #1 \right)}

\newcommand{\dmain}{\cX}
\newcommand{\range}{\cY}
\newcommand{\kdrange}{\range_{\ms,\rd}^{\text{disc}}}

  \newcommand{\cA}{\ensuremath{{\mathcal A}}\xspace}

  \newcommand{\cG}{\ensuremath{{\mathcal G}}\xspace}

  \newcommand{\cM}{\ensuremath{{\mathcal M}}\xspace}

  \newcommand{\cS}{\ensuremath{{\mathcal S}}\xspace}

  \newcommand{\cX}{\ensuremath{{\mathcal X}}\xspace}
  \newcommand{\cY}{\ensuremath{{\mathcal Y}}\xspace}

\usepackage[
    backend=biber,
    citestyle=alphabetic,   %
    bibstyle=alphabetic,    %
    sorting=nyt,
    maxbibnames=20,
    natbib=true,            %
    uniquelist=false
]{biblatex}
\usepackage{thmtools,thm-restate}

\Crefname{algocf}{Algorithm}{Algorithms}

\begin{document}

\maketitle

\begin{abstract}
Graph property testing asks whether a massive graph satisfies a given property, or is far from doing so, using only a sublinear number of queries to the graph. Since property testers typically inspect only a small, randomly sampled portion of the input, they appear naturally compatible with differential privacy and privacy amplification by subsampling. Despite this, few results link these two fields. We initiate a systematic study of differentially private graph property testing with the goal of designing efficient testers with formal privacy guarantees in the \emph{dense} and \emph{bounded-degree} graph models. We develop new privacy amplification theorems for several widely used graph-sampling procedures such as \emph{induced subgraph sampling}, \emph{random walks} and \emph{\texorpdfstring{$k$}{k}-disc sampling}. We then leverage these privacy amplification techniques to design a \emph{private canonical tester} in the dense graph model, as well as \emph{private bipartiteness testers} and \emph{subgraph freeness testers} in the dense and bounded-degree graph models. Finally, using the new privacy amplification theorem for \texorpdfstring{$k$}{k}-disc sampling, we prove that every property of \emph{hyperfinite} graphs is privately testable. The resulting query complexities of our private testers are comparable to those of their non-private counterparts.

\end{abstract}

\section{Introduction}

Graphs frequently arise as a natural model or abstraction of private data---consider social networks, computer networks, and relational databases, such as records of transactions or phone calls. These graphs are often massive, far beyond what one could process in linear time or space. This constraint has motivated the study of sublinear algorithms for \emph{testing graph properties}. Despite the deep, ongoing study of graph property testing and the widespread use of graphs as representations of sensitive data, little is known about the privacy of graph testing algorithms. In this work, we initiate a principled study of private graph property testing. 

The scale and scope of data collection today have led to huge innovations in the study of formal privacy guarantees, particularly \emph{differential~privacy~(DP)}~\cite{dwork2006}. %
Differential privacy on graphs has been extensively studied (see~\cite{RaskhodnikovaS16} for a survey), and can be broken down into two categories. Privacy can be established at the level of nodes (node-DP), typically representing individuals, where privacy is with respect to a differing node (including all edges incident to that node), or at the level of edges (edge-DP), typically representing relations between nodes, where privacy is with respect to a differing edge. Differential privacy for graphs can be formalized as follows.

\begin{definition}[Differential~privacy~{\cite{dwork2006}}]\label{def:DP}
  Let $\mathcal{G}_n$ denote the set of all $n$-vertex graphs, $\pdelta\in[0,1]$ and $\priv>0$. A randomized algorithm $\cA:\cG_n\to\cY$ is said to satisfy $(\eps,\delta)$-node-differential privacy, or node-DP in short (resp. edge-DP) if for every pair of node-neighboring (resp. edge-neighboring)\footnote{Two graphs $G_1=(V, E_1)$, $G_2=(V, E_2)$ are said to be \emph{node-neighboring}, denoted by $G_1\sim_v G_2$, if there exists a vertex $v\in V$ such that $E_1(V\setminus\{v\})=E_2(V\setminus\{v\})$. Similarly,  graphs $G_1$ and $G_2$ are called \emph{edge-neighboring} i.e., $G_1 \sim_e G_2$ if there exists an edge $e$ such that $E_1 \setminus{\{e\}}=E_2 \setminus{\{e\}}$.} graphs $G_1,G_2 \in \mathcal{G}_n$, and for all $\cS\subseteq\cY$, we have that $\Pr[\cA(G_1) \in \cS ] \leq e^\eps \Pr [\cA(G_2) \in \cS] + \delta$. When $\delta=0$, we simply say that the algorithm is $\eps$-DP, or that it satisfies pure differential privacy.
\end{definition}

Graph property testing is the problem of deciding whether a graph satisfies a graph property %
or is $\dst$-far from doing so, under an appropriate choice of distance for some parameter $\dst \in (0,1]$. Let us first formally define graph properties. 

\begin{definition}[Graph properties]\label{def:graph-property}
    A graph property is a set of graphs that is closed under graph isomorphism. More formally, $\property$ is a graph property if, for every graph $G=(V, E)$ and every bijection $\sigma:V\to V'$, $G \in \property$ if and only if $\sigma(G) \in \property$ where $\sigma(G)\defeq (V', \{\{\sigma(u), \sigma(v)\}: \{u,v\} \in E\} )$. 
\end{definition}
\noindent In this work we will only consider graphs with vertex set $V=[\ns]\defeq \{1,\dots,\ns\}$. 

Now we are ready to define the notion of graph property tester for a property \property, an algorithm that is given oracle access to a graph $G=(V,E)$ and has to decide whether $G$ is in \property, or far from being in \property.
\begin{definition}[Graph property tester]\label{def:graph-property-tester}
  A randomized algorithm is said to be a two-sided error \emph{$\dst$-property tester} $\tester$ for a graph property $\property$ if given oracle access to an unknown graph $G$ as input and a parameter $\dst \in (0,1]$, it satisfies the following two conditions of \emph{completeness} and \emph{soundness}.
  \begin{enumerate}
      \item \textbf{Completeness:} If $G\in\property$ then $\tester$ accepts with probability at least $2/3$.
      \item \textbf{Soundness:} If $G$ is $\dst$-far from every graph in \property then $\tester$ rejects with probability at least $2/3$. 
  \end{enumerate}
The distance used will depend on the model. If $\tester$ accepts every graph in \property with probability 1, then we say that it has one-sided error; otherwise, we say it has two-sided error. {The total number of queries performed by $\tester$ is known as its \emph{query complexity}. 
The query complexity of testing the property $\property$ is the minimum over the query complexities of all testers $\tester$ for $\property$.  
The notation $q(\ns,\dst,1/3)$ denotes the query complexity of $\tester$ when the unknown graph has $n$ vertices, the proximity parameter is $\dst$ and the failure probability of $\tester$ is at most $1/3$.}
\end{definition}

Graph property testing was initiated in the seminal work of \cite{goldreich1998property}. Over the past two decades, this area has been studied extensively, yielding a wide range of techniques (see the books~\cite{goldreich2017introduction,bhattacharyya2022property} and the surveys~\cite{DBLP:journals/fttcs/Ron09, DBLP:conf/propertytesting/CzumajS10,DBLP:journals/siamdm/RubinfeldS11} for more references).

A natural reason to expect some compatibility between graph property testing and differential privacy is that graph property testers are inherently sublinear. These algorithms typically inspect only a small, randomly selected portion of the input graph. This limited access naturally leads one to consider \emph{privacy amplification by subsampling}. When a DP mechanism is applied to a randomly sampled portion of a graph, a particular edge or vertex can only influence the result in the case that it is sampled. This connection is evidenced by the development of new private sublinear algorithms~\cite{blocki2022, BlockiGMZ23, BlockiGM21}. On the other hand, it is known that not all sublinear algorithms can be made differentially private \cite{BlockiFGM25}, thus revealing a dichotomy that motivates a more careful study of when sublinear access and the associated sampling procedures yield useful privacy guarantees.

Graph property testing provides a particularly natural setting to investigate this dichotomy. Since a property tester ultimately outputs only a single binary decision --- ``Yes" if the input graph satisfies the property and ``No'' when it is far from the property --- the classical \emph{Randomized Response}\footnote{Instead of releasing the tester's decision directly, the true answer is reported with high probability and flipped with some probability chosen to guarantee differential privacy, see \Cref{def:randomized-response}.}~\cite{kasiviswanathan2011} is a natural privacy mechanism. In particular, one can first run a non-private tester and then apply Randomized Response to its output. However, obtaining a strong privacy guarantee in this way alone may require a non-negligible probability of flipping the tester's decision, thereby degrading its accuracy. This motivates combining Randomized Response with privacy amplification arising from the tester's sampling procedure. 

To this end, the heart of our results is a collection of \emph{privacy amplification theorems for graphs}, tailored to sampling schemes commonly used in sublinear graph algorithms. We use these amplification results to design $\varepsilon$-node DP (and edge-DP) graph property testers in the dense and bounded-degree models. However, the amplification theorems themselves apply more broadly and may be useful for other sublinear graph algorithms.

\paragraph{Concurrent Work. }
To the best of our knowledge, the recent work of Dwork and Tankala~\cite{dwork2026efficientprivatepropertytesting} was the first to explicitly mention graph property testing under differential privacy. There, the authors were interested in identifying \emph{regularity templates} in the context of differential privacy, a combinatorial object of fundamental importance in graph property testing~\cite{alon2009combinatorial}. Their analysis arose from a more detailed study of indistinguishability. The authors used \emph{privacy amplification by subsampling}~\cite{steinke2022,balle2020a} in order to demonstrate that their algorithm can be made to satisfy $O(1/\ns)$-node differential privacy for any choice of input parameters. 

In terms of property testing, they extend the combinatorial characterization of constant-query testable properties in the dense graph model~\cite{alon2009combinatorial}, and show that the same set of properties is constant-query testable with node-level DP. {Extending this characterization, our work studies the canonical tester under node and edge DP, demonstrating that even non-constant query complexity properties are privately testable. We also discuss the bounded-degree setting at length, which was not discussed previously. Since there is no analogue to the canonical tester in this setting, we instead analyze two of the most common sampling schemes, and state our results as general applications of these schemes.}

\subsection{Our Results}

Our main results are privacy amplification theorems for different graph sampling schemes including induced subgraph sampling in the dense graph model, and random walks and $k$-discs (small breadth-first searches around randomly selected root nodes) in the bounded-degree model, which lead to $\eps$-DP property testers in the dense and bounded-degree model.  
Let us first define the notion of a private tester.

\begin{definition}[Private tester]
Consider a $\dst$-tester with query complexity $q$ for a property \property as defined in \Cref{def:graph-property-tester}. An $\priv$-private $\dst$-tester with query complexity $q'$ for \property is a tester that also satisfies $\priv$-differential privacy as defined in \cref{def:DP}.     
\end{definition}

{Before stating our main results, we first observe that property testers with one-sided error are incompatible with pure $\priv$-differential privacy. Consider a tester $\tester$ for a property $\property$ with one-sided error, by definition,  if $x \in \property$, then $\Pr[ \tester \text{ rejects } x] = 0$. If $\tester$ satisfies $\eps$-DP, then for any neighboring $x$ and $x'$, $\Pr[\tester \text{ rejects } x'] \leq e^{\eps} \Pr[ \tester \text{ rejects } x] = 0$. By transitivity, $\tester$ has to accept on its whole domain\footnote{See a formal proof in \cref{sec:pure-dp}.}. Therefore, a private property tester for any non-trivial property would need to relax either the one-sided error guarantee or the pure DP guarantee. In this paper, we choose to present our results with the strongest privacy guarantees, $\eps$-differential privacy. In particular, our private testers satisfy $\eps$-DP and are of two-sided error\footnote{Technically, our resulting private testers can be considered as ``almost one-sided''. This is because the completeness error (which needs to be 0 for one-sided error) is very small, approximately $o(1)$. See \cref{rem:almost-one-sided} for details.}. 
}

Many of our results are presented with respect to {reasonable lower bounds on $\priv$ (typically of the form $\priv\gtrsim \log(1 + \frac{1}{\ns})$),} above which there exist private testers with query complexities within a constant factor of their non-private counterparts. {One should note that below this value of $\priv$, {our testers are still private}, but the error induced by privacy will be greater than that of the non-private tester.} {In particular, choosing $\priv$ to be a constant in our algorithms yields  $\bigO{1/\ns}$-DP testers.}
The formal theorem statements include all relevant details about parameter regimes.

The rest of this section is organized as follows: We first state our privacy amplification theorems for the graph sampling schemes specific to the graph access model, and then state the guarantees of the private property testers we obtain by applying the amplification theorems therein. Throughout, we use $O_{\dst,\priv}(\cdot)$ to suppress multiplicative factors that may depend on the proximity parameter $\dst$, and privacy parameter $\eps$, and we use \(\widetilde O(\cdot)\) to suppress polylogarithmic factors in the relevant problem parameters.

\paragraph{Private property testing in the dense graph model. }

Let us start with our results on dense graphs. The input is an $n$-vertex graph accessed through \emph{adjacency queries}: given a pair of vertices $\{u,v\}$, the algorithm learns whether $\{u,v\}\in E$ or not. Distance from a graph property is measured relative to the $\Theta(n^2)$ possible edges.

Consider the \emph{induced subgraph sampling} approach, where we sample a set of (smaller) subgraphs $X_1, \ldots, X_\ell$ from the underlying unknown graph $G$, and based on the sampled subgraphs, decide whether $G$ satisfies some property of interest. This is one of the standard techniques to design testers in the dense graph model. We state privacy amplification results for node and edge-neighboring graphs. 

Note that privacy amplification by induced subgraph sampling for node-DP is implied by~\cite[Lemma 2.16]{dwork2026efficientprivatepropertytesting}.
We prove the following nontrivial extension to edge-DP. 
\begin{theorem}[Edge-DP privacy amplification, Informal version of \Cref{thm:induced-edge-amplification}]\label{thm:induced-edge-amplificationintro}
Let $G=([\ns],E)$ be a dense graph and $\pmech$ be an $(\priv, \pdelta)$\emph{-edge differentially private} algorithm that takes as input a sample of $\ms$~subgraphs of size $h$. Let $\eta = \frac{\ms h(h-1)}{\ns(\ns-1)}$. Applying \pmech to random subgraph samples of $G$ yields an $(\priv', \pdelta')$-edge DP algorithm $\pmech'$ where $
\priv' = \log \left( 1 + \eta(e^{2(h-1)\priv} - 1) \right) $, and, $
\pdelta' = \eta \frac{e^{2(h-1)\priv} - 1}{e^{\priv} - 1}\pdelta.$
\end{theorem}
Unlike in the node-DP case, the distribution of subgraph samples that \emph{do} contain some edge is not necessarily close to the distribution of subgraph samples that \emph{do not} contain that edge. In fact, the presence or absence of an edge in a sample is wholly decided by whether its vertex endpoints are sampled. As such, samples that do contain a particular edge are unlikely to be edge-neighboring to those that do not, thus requiring a more involved and delicate analysis.
In the sequel, we use both the edge and node-privacy amplification theorems to design edge and node-DP property testers in the dense graph model. 

\Cref{thm:induced-edge-amplificationintro} is extremely general as stated, proving amplification for $(\priv,\pdelta)$-differentially private algorithms, where throughout we only use pure-DP. Furthermore, it interpolates nicely between the parameters $\ms$ and $h$ that represent the number and size of the sampled subgraphs. 

To highlight the tradeoff between these parameters, we consider first the canonical tester. The \emph{canonical testing} algorithm of \cite{goldreich2003three} is one of the most remarkable results of dense graph property testing.~\citet{goldreich2003three} proved that any $q$-query tester in the dense graph model can be transformed into a canonical non-adaptive tester that samples $O(q)$ vertices and queries the entire induced subgraph, incurring only a quadratic blow-up to $O(q^2)$ queries.
~As a first application of our edge-DP privacy amplification technique, we design an edge-private canonical tester for dense graphs. Note that the corresponding statement for node-differential privacy follows from standard privacy amplification by subsampling for databases and we state it below for completeness. %
\begin{theorem}[Private canonical tester]\label{thm:dpcanonicaltesterintro}
Consider any graph property $\prop$. Let  \tester be an arbitrary tester for \prop with query complexity $q$. Then,
\begin{itemize}
    \item {(Informal version of \Cref{thm:dpcanonicaltester})} For all $\priv\geq \log(1 + 4\frac{q^2}{\ns^2})$, there is an $\priv$-\emph{edge DP} tester for $\prop$ with query complexity {at most} $2q^2$. 
    \item (Informal version of~\cref{cor:node-dp-canonical}) For all $\priv \geq \log(1 + 4\frac{q}{\ns})$, there is an $\priv$-\emph{node DP} tester for $\prop$ with query complexity {at most} $2q^2$. 
\end{itemize}
\end{theorem}

To the best of our knowledge, these are the first edge and node-DP canonical testers for graph property testing. 
Notably this theorem implies the existence of private testers with query complexities differing by only a constant factor from their non-private counterparts.

As an application of this tester, we consider the problem of \emph{bipartiteness testing} in the dense graph model. In this problem, given query access to an unknown dense graph $G$, the goal is to distinguish whether $G$ is bipartite, or $G$ is $\dst$-far from bipartite, for some parameter $\dst \in (0,1]$. There exists a canonical tester for bipartiteness that samples a small, fixed number of vertices, and then checks whether the induced subgraph is bipartite. Applying~\cref{thm:dpcanonicaltesterintro}, we see that bipartiteness is privately testable in the dense graph model.
We will also later consider the harder problem of bipartiteness testing in the bounded-degree model.

The private canonical tester (\Cref{thm:dpcanonicaltesterintro}) samples one large subgraph, setting $\ms=1$ and $h=q$ in the application of~\cref{thm:induced-edge-amplificationintro}. On the other end of the spectrum are problems where we might perform many small, or constant-sized tests of the graph, setting $\ms>1$ and $h$ to be constant. This approach to property testing is exemplified by the problem of testing \emph{$H$-freeness}. Here, given query access to an unknown dense graph $G$ and the complete description of a small subgraph $H$, the goal is to distinguish whether $G$ does not contain any copy of $H$ (up to isomorphism, see \Cref{defi:subgraphfreness}) or is far from being $H$-free. 
{Similar to the work of \cite{DBLP:journals/jal/AlonDLRY94,DBLP:journals/combinatorica/AlonFKS00}}, we use the well-studied \emph{graph removal lemmas} from \cite{ConlonF13}, a versatile and widely used tool in combinatorics. Combined with our induced subgraph sampling technique from \cref{thm:induced-edge-amplificationintro}, this gives us a private tester for $H$-freeness testing. 

We use the fact that a dense graph that is $\dst$-far from being $H$-free must contain at least $\prox(\dst,H)\cdot \ns^{|V(H)|}$ copies of $H$, for some parameter $\prox$ (see~\cref{lem:graph-removal} for more details). Importantly, $\prox$ can be enormous, such as in the case of triangle-freeness testing, while $h$ in the same example is only 3.

\begin{theorem}[Private $H$-freeness tester, Informal version of \cref{thm:private-POT-h-freeness}]\label{private-POT-h-freenessintro}
Fix a subgraph $H$, with $h$ vertices and distance parameter $\dst\in(0,1]$. 
For all $\priv=\bigOme{\log(1 + \frac{h^2}{\prox\ns^2})}$, there exists $\priv$-{edge}-DP $H$-freeness $\dst$-tester with query complexity $O\!\left(h^{2}/\prox\right)$. 
\end{theorem}
\noindent In comparison, applying the canonical tester would increase the query complexity to $h^2/\prox^2$ and require a lower bound of $\priv$ of the order of $\bigOme{\log(1 + \frac{h^2}{\prox^2\ns^2})}$.

Interestingly, our $H$-freeness tester is an application of a \emph{proximity-oblivious} tester, in the sense that the final $H$-freeness tester is composed of applying a basic tester multiple times, and the basic tester is independent (aka. oblivious) of the proximity parameter $\dst$. These testers seem particularly amenable to privacy amplification, and we will see another application of this framework in the bounded-degree setting.

\paragraph{Private property testing in the bounded-degree model. }We state our results for bounded-degree graphs next. The input is an $n$-vertex graph of maximum degree at most $d$, accessed through \emph{neighbor queries}: given a vertex $v$ and an index $i\in[d]$, the algorithm returns the $i$th neighbor of $v$ (according to some fixed ordering of the neighbors). Distance from a graph property is measured relative to the $O(dn)$ possible edge incidences. Since there is no canonical tester in this model, we study two algorithmic techniques widely used to design testers in this model: random walk based testers and BFS (breadth first search)-based testers. 

We start by discussing the design of private random walk based testers in the bounded-degree model. The random walk approach is ubiquitous in designing testers for bounded-degree graphs since the graphs can only be explored locally via the neighbors of the vertices. 
{Although random walks are generated adaptively, we mainly reason about the transcript generated by these random walks, which is essentially the list of all the vertices in the order they were visited by a sequence of random walks. }
Unlike the privacy amplification result for induced subgraph sampling in the dense graph model, where our result applied more generally, our privacy amplification result for random walks applies only to algorithms that output a binary decision, such as graph property testers. The main challenge towards proving a more general statement is that a single edge change can cause the subsequent random walk transcripts to diverge significantly. Without additional assumptions on the input graphs or random walks, it is hard to bound the possible scale of divergence in a reasonable manner for privacy amplification. Restricting to binary decisions lets us avoid comparing these transcripts directly, consequently, we state the following result in terms of a tester. We note that this result holds only for edge-DP\footnote{Though, in the $\maxdeg$-bounded degree model, $\priv$-edge differential privacy combined with group-privacy~(\cref{def:group-privacy}) implies $\priv'$-node DP where $\priv' = \maxdeg\cdot \priv$.}, vs. our other amplification results which hold for both edge and node-DP.

\begin{theorem}[Privacy amplification of random walks, Informal version of \Cref{thm:dpamplificationrandwalk}]\label{thm:dpamplificationrandwalkintro}
Let $G=([\ns], E)$ be a bounded degree graph. Suppose $\pmech$ is the application of randomized response with privacy parameter $\priv$ (see~\cref{def:randomized-response}) to the output of a tester {with query complexity $\tsize$} that samples random walks.\footnote{For our purpose, $\tsize \leq n$ since we can learn the graph completely with $O(n)$ queries.} Then running \pmech on $G$ yields an $\priv'$-edge DP algorithm, where 
\[
\priv' = \log \left( 1 + \frac{2\tsize}{\ns}(e^\priv - 1) \right).
\]
\end{theorem}

As an application of our privacy amplification result for regularized random walks, we revisit the problem of bipartiteness testing, this time in the bounded-degree model. Bipartiteness testing in this model requires $\widetilde{\Theta}(\sqrt{\ns})$ queries~\cite{goldreich1999a, goldreich2002property}. This is in contrast to the dense graph model, where the query complexity is independent of $\ns$. The bipartiteness tester is notable for performing many long random walks in $G$. Using our privacy amplification result, we obtain the first private bipartiteness tester for bounded-degree graphs.

\begin{theorem}[Private bipartiteness tester in bounded-degree graphs, Informal version of \Cref{thm:dpbipartitebounded}]\label{thm:dpbipartiteboundedintro}
For all $\priv=\widetilde{\Omega}_{\dst,\maxdeg}\left( {\log(1 + 1/\sqrt{\ns})} \right)$, there exists $\priv$-edge DP bipartiteness tester with query complexity $\widetilde{O}_{\dst,\maxdeg}(\sqrt{\ns})$. Notably this is within constant factors of the non-private query complexity bound.
\end{theorem}

{We use the bipartiteness tester of~\cite{goldreich1999a,goldreich2002property} and not the state-of-the-art tester of~\cite{fei2026}. This is due to an incompatibility between our privacy amplification result and the specific kind of random walk they use.
For a full discussion, we refer interested readers to~\cref{sec:random-walks}.} 

We end by discussing the BFS-based graph testers. Here, the approach is to start BFS of sufficient depth from a few randomly sampled vertices and decide based on the explored subgraphs. This is sometimes known as \emph{$k$-disc sampling}. We first design a privacy amplification technique for $k$-disc sampling. {We note that similar to our induced subgraph sampling privacy amplification result, the following privacy amplification theorem applies more generally to any $(\eps,\delta)$-DP algorithm that uses the $k$-disc sampling scheme. 
}

\begin{theorem}[Privacy amplification for \texorpdfstring{$k$}{k}-disc sampling, Informal version of \Cref{thm:kdiscamplification}]\label{thm:kdiscamplificationintro}
Let $G=([\ns],E)$ be a $\maxdeg$-bounded-degree graph. Suppose $\pmech$ is an $(\priv, \pdelta)$\emph{-differentially private} algorithm that samples $\ms$-tuples of $\rd$-discs.
Then running $\pmech$ on $G$ yields an $(\priv', \pdelta')$-\emph{edge and node} DP algorithm, where 
\[
\priv'=K\log\left( 1 + \frac{\ms}{\ns}(e^{ \priv }-1) \right) \qquad \qquad \pdelta'=\pdelta\cdot K \frac{\ms}{\ns}\left(  \frac{\ms}{\ns}(e^{ \priv }-1)+1 \right)^{K},
\]
and $K\leq 2\maxdeg^{\rd}+1$.
\end{theorem}

As an application of our privacy amplification technique for $k$-disc sampling, we revisit the problem of $H$-freeness testing, but in the bounded-degree model. Since we can design a proximity-oblivious tester for $H$-freeness testing using the $k$-disc sampling technique, we directly obtain a private $H$-freeness tester in the bounded-degree model.

\begin{theorem}[Private \texorpdfstring{$H$}{H}-freeness tester in bounded-degree graphs, Informal version of \Cref{thm:dp-h-freeness}]\label{thm:dp-h-freenessintro}
Let $G=([\ns],E)$ be a $\maxdeg$-bounded degree graph, $H$ be a graph with $h$ vertices and radius at most $\rd$ (see~\cref{def:k-disc}). Let $K=2\maxdeg^\rd+1$ be an upper bound on the number of edges or nodes in one disc. For all $\priv\geq 8K/(\dst\ns)$, there exist $\priv$-{node and edge} DP $H$-freeness \dst-testers with query complexity $8K/\dst$. This is within a constant factor of the query complexity of the non-private tester.
\end{theorem}

Besides the problem of $H$-freeness testing, we also consider \emph{hyperfinite} graphs. A remarkable result of \cite{newman2013every} proved that any property of hyperfinite graphs can be (non-privately) tested by performing a number of queries that is independent of the size of the input graph, such properties are known as \emph{testable} properties. Interestingly, \cite{newman2013every} uses a $\rd$-disc-sampling based method to design their tester. By virtue of our privacy amplification result for $\rd$-disc-sampling based testers, we show that every property of hyperfinite graphs is privately testable. {We note that the following result shows the existence of a tester that is private regardless of the input graph being hyperfinite\footnote{For the bounded degree model, we assume that the input graphs always satisfy the bounded degree, and we argue about privacy with this assumption in mind. See \Cref{sec:discussion} for a more detailed discussion.}, but is only accurate if the input graph is hyperfinite. Unlike the non-private result, we do not check if the input graph is hyperfinite and thus our result is slightly weaker than the non-private variant.}

\begin{theorem}[{Private} testability of hyperfinite graph properties, Informal version of \Cref{thm:dphyperfinitetestable}]\label{thm:dphyperfinitetestableintro}

{Let $G=([\ns],E)$ be a $\maxdeg$-bounded degree graph and $\priv>\widetilde{\Omega}_{\maxdeg,\rho}({1/\ns})$. If $G$ is $\rho$-hyperfinite, then any property of $G$ is $\priv$-node and edge-DP testable with constant failure probability. If $G$ is not hyperfinite, then the algorithm is private, but does not have accuracy guarantees.}

\end{theorem}

\subsection{Technical Overview}

{Our private testers are designed by first running the corresponding non-private testers and then applying Randomized Response $RR_\eps$ (\Cref{def:randomized-response}) to the output (See Algorithm~\ref{alg:meta-tester-private}). Since $RR_\eps$ satisfies $\eps$-DP, the resulting private testers also satisfy $\eps$-DP. However, the accuracy of these testers may not necessarily compare to the accuracy of non-private testers. For this reason, we study privacy amplification, so that our resulting private testers satisfy $\eps'$-DP where $\eps' \ll \eps$, thus resulting in comparable accuracy to non-private testers. }

\begin{algorithm}\label{alg:meta-tester-private}
\caption{Private tester template for graphs}
\KwIn{Query access to an unknown graph $G$, $\dst$-tester $\mathcal{T}$ for property $\property$, privacy parameter $\eps$.}
\KwOut{\texttt{Accept} or \texttt{Reject}}

Run the non-private tester $\mathcal{T}$ for $\property$.

Apply $RR_\eps$ to the output of $\mathcal{T}$ and return the output.
\end{algorithm}

Our main technical contribution is a collection of privacy amplification theorems tailored to sampling procedures that arise naturally in sublinear graph algorithms. Note that privacy amplification by subsampling has been considered previously for graph algorithms, but only for sampling schemes that permit simple analogies to existing amplification results on datasets. These include an edge-privacy amplification by edge-subsampling lemma of~\cite{zhou2025}, and a node-privacy amplification by node-subsampling lemma of~\cite{dwork2026efficientprivatepropertytesting}. {Our work, in contrast, establishes privacy amplification for graph-specific sampling procedures where the neighboring relation on input graphs does not translate directly to the sampled objects, including adaptive procedures such as random walks.} 

\paragraph{Classical vs. graph privacy amplification. }Classical privacy amplification by subsampling on datasets relies on a relatively simple relationship between the neighboring relation on the original dataset and that on the subsample --- neighboring datasets differ in one record, and the subsampled datasets differ only if that record is selected. More complex cases can include considering the probability that a record is selected multiple times while sampling with replacement. The privacy-amplification decomposition of~\cite{balle2020a} (stated as \emph{advanced joint convexity} in \Cref{advanced-joint-convexity}) formalizes this by decomposing the subsampling distributions on neighboring datasets into mixtures comprised of a \emph{common conditional distribution} where the differing record \emph{is not} sampled, together with \emph{residual conditional distributions} corresponding to the event that the differing record \emph{is} sampled. 
By necessity, the outputs of a private algorithm run on two neighboring datasets must be indistinguishable in the instance that the differing record is not sampled. Then intuitively the overall privacy loss should be proportional to the probability with which we sample the said element.

For graph sampling procedures, changing a single edge or node can affect the sampling distribution in unexpected, or hard to quantify ways; consider an edge that connects two otherwise disconnected components in a graph, removing this edge could change the trajectory of a very large number of random walks. Our analyses therefore identify an appropriate mixture decomposition into common and residual distributions for each graph sampling primitive and establish the additional relationships between these distributions (through couplings and group privacy) needed to invoke the same amplification framework of~\cite{balle2020a}. %

\paragraph{Induced subgraph sampling. }Consider first the dense-graph setting, where we sample $\ms$ disjoint sets or blocks of $h$ vertices and observe the corresponding induced subgraphs. For two edge-neighboring $\ns$-vertex graphs $G,G'$ that differ on an edge $e=\{x_1,x_2\}$, the induced-subgraph samples are identical unless $e$ is ``exposed'', i.e., both endpoints of $e$ appear in the same sampled block. This event occurs with probability $\eta=\frac{\ms h(h-1)}{n(n-1)}$. For simplicity of exposition, we will say that $e$ is present in $G$, and absent in $G'$.
As in the classical setting, conditioning on whether $e$ is exposed yields (1) a common distribution corresponding to non-exposure of $e$, and (2) two residual distributions corresponding to samples that expose the presence of $e$ in $G$ and those same samples that expose the absence of $e$ in $G'$.
The two residual distributions essentially couple themselves, for every sample that exposes $e$ in $G$, the same choice of vertices will expose the absence of $e$ in $G'$. 
The additional difficulty in applying the privacy-amplification framework of~\cite{balle2020a} to graphs arises when we attempt to bound the distance between samples from the common conditional distribution to those in either of the residual distributions.

Unlike in classical subsampling, these two distributions need not be edge-neighboring, i.e., differing by a single edge. Conditioned on $e=\{x_1,x_2\}$ being exposed, a sampled block from one of the residuals contains both $x_1$ and $x_2$, whereas its nearest neighbor in the common conditional distribution may replace $x_2$ by another vertex. This can change all $h-1$ edges adjacent to $x_2$, rather than only the edge $e$. 
The key observation is that the underlying block configurations of these conditional distributions differ only in the placement of the endpoints $x_1$ and $x_2$. 
This allows us to construct a \emph{distance-bounded coupling} (\Cref{def:dcc}) between an exposed configuration with a non-exposed one by moving or replacing these endpoints, while changing only a bounded number of edges in the resulting induced subgraphs. 
We then invoke group privacy on this distance-bounded coupling to obtain the desired amplification guarantee in \Cref{thm:induced-edge-amplificationintro}. 

We emphasize that the privacy amplification result for induced subgraph sampling applies more generally to any DP algorithm employing this sampling scheme, and is not restricted to DP algorithms that output binary decisions, which would be sufficient for our paper, as we are mainly interested in designing DP graph property testers. 
{In our paper, we design the private canonical tester (\Cref{thm:dpcanonicaltesterintro}), as well as the $H$-freeness tester (\Cref{private-POT-h-freenessintro}), using this privacy amplification result.}

\paragraph{Random walk sampling. } The most commonly used form of random walk in graph property testing is the \emph{lazy regularized random walk}, often referred to as a lazy random walk, and it has the property of being \emph{doubly-stochastic}. Although a random walk is generated adaptively, we will {reason about} complete transcripts, thought of as the list of all the vertices in the order they were visited by a sequence of random walks. This allows us to view random walk sampling as inducing a distribution over fixed transcript objects, and to compare these distributions rather than reason about the adaptive choices step by step. Proving general privacy amplification guarantees for private algorithms that implement a random walk sampling can be challenging. To see this, consider two edge-neighboring $\ns$-vertex graphs $G,G'$ that differ on an edge $e=\{x_1,x_2\}$. 
Similar to the case of induced subgraph sampling, the privacy-amplification decomposition from~\cite{balle2020a} yields (1) a common conditional distribution corresponding to non-exposure of $e$, and (2) two residual conditional distributions corresponding to exposure of $e$ in $G$ and $G'$, respectively. Here, $e$ being exposed means either endpoint $x_1$ or $x_2$ appears in the random walk transcript. Unlike in induced subgraph sampling, however, even relating the two residual conditional distributions (from (2)) is not straightforward, as they need not be close --- once the walk reaches $x_1$ or $x_2$, the walks on $G$ and $G'$ may diverge for the remainder of the execution. Therefore, the corresponding transcripts may differ in many positions, so we cannot directly relate the two conditional residual distributions by a small distance bound as in the induced-subgraph case. Without a good upper bound on the distance between these transcript distributions, it seems hard to prove a general privacy amplification result, unless we make additional assumptions about the structure of the graphs or random walks themselves. Since we are primarily interested in designing private property testers in our work, we restrict ourselves to proving privacy amplification guarantees for edge-DP algorithms that output \emph{binary decisions}. 

In this case, we consider random walk transcripts that have been mapped to a bit by the non-private tester, and then apply Randomized Response to the resulting decision. Crucially, Randomized Response is $\varepsilon$-DP between \emph{any} two possible binary inputs, regardless of how different the underlying transcripts producing those inputs may be. This lets us handle both types of comparisons --- between the two residual conditional distributions corresponding to exposure of $e$ in $G$ and $G'$ (from (2)), and between either residual conditional distribution and the common conditional distribution (from (1)) --- required by the privacy amplification decomposition uniformly. Once these comparisons are handled by Randomized Response, it remains only to bound the weight of the residual distributions, which is controlled by the probability that the differing edge $e$ is exposed. For a lazy regularized random walk started from a uniformly random vertex, every subsequent position is also uniformly distributed over the vertices (because it is doubly stochastic), allowing us to bound the probability of reaching either endpoint of $e$ by the fraction of the graph explored by the walks. Combining this probability bound with the privacy-amplification framework of~\cite{balle2020a} yields our amplification guarantee for random walk based property testers in \Cref{thm:dpamplificationrandwalkintro}. {In our paper, we use the privacy amplification result for random walks to design a private tester for bipartiteness testing in the bounded-degree model (\Cref{thm:dpbipartiteboundedintro}). }

\paragraph{\texorpdfstring{$k$}{k}-disc sampling. }In the bounded-degree model, testers are often based on a small number of local explorations of the graph to a constant depth. %
A $k$-disc sampler chooses uniformly random roots and explores their radius-$k$ neighborhoods by breadth-first search. A single change to the underlying graph can therefore affect many possible $k$-discs, i.e., changing one node or edge may alter every disc whose root lies sufficiently close to that modification. Our analysis views the collection of all rooted $k$-discs of the graph as an auxiliary database. For two node (or edge)-neighboring bounded-degree graphs, all but a bounded number $K$ of these discs are identical, where $K$ is controlled by the maximum size of a radius-$k$ neighborhood. Sampling $r$ random roots is then exactly subsampling $r$ elements from this auxiliary database. For each differing entry, the classical subsampling privacy amplification decomposition directly gives a common conditional distribution when that entry is not sampled and residual conditional distributions when it is sampled. We can therefore apply classical privacy amplification at the level of the auxiliary database, and then use group privacy to account for the fact that one graph modification may change up to $K$ database entries. Observe that, unlike induced subgraph sampling and random-walk sampling, no new relationship between the common and residual conditional distributions needs to be established; the main task is instead to bound how many $k$-discs can be affected by a single graph modification. {We note that similar to the induced subgraph sampling result, our privacy amplification theorem in \Cref{thm:kdiscamplificationintro} applies more generally to private algorithms that use $k$-disc sampling as a primitive. In our paper, we use this result to design a private tester for $H$-freeness (\Cref{thm:dp-h-freenessintro}) and a general tester for hyperfinite graphs in the bounded-degree model (\Cref{thm:dphyperfinitetestableintro}). }

\subsection{Discussion}\label{sec:discussion}

\paragraph{Privacy in the bounded-degree model. }In the bounded-degree model, the maximum-degree bound $d$ is essential to our $k$-disc privacy amplification theorem. 
It is an important practice in differential privacy to avoid assuming that the inputs satisfy a promise, such as having bounded maximum degree, mainly because algorithms having such restrictions can exhibit blatant failures of differential privacy on graphs that violate this promise. Typically, in such settings, the goal is to guarantee privacy on all arbitrary input graphs, but guarantee accuracy only on those inputs that satisfy the promise~\cite{JainSW24}. 
A standard DP technique for bounded-degree graphs was introduced by~\cite{blocki2013a, kasiviswanathan2013}, where they constructed projections from arbitrary graphs into a bounded-degree hypothesis class and then ran the desired query on the projected graph. This ensures privacy for all input graphs, but accuracy of the desired query for graphs satisfying the degree bound. Such a preprocessing approach is not directly compatible with our objective of preserving the sublinear complexity of the underlying property tester, as even the most efficient DP projections require reading the entire graph, thus forfeiting the sublinear access advantage that motivates our setting. It is an interesting open question as to whether such DP projections can be designed in sublinear time, i.e., without reading the entire graph as input. 

\paragraph{More private property testers from privacy amplification. }In our work, we highlighted private property testers for some fundamental properties that can be obtained from our privacy amplification theorems. We believe our amplification theorems can be used to obtain private property testers for more properties and {merely state them here}. {Consider the problem of \emph{expansion testing}, where the goal is to distinguish whether an unknown graph is expanding (with respect to a sufficient parameter) or far from all such graphs.}
{The expansion tester of~\cite{czumaj2007} {and the more recent improvement by~\cite{DBLP:journals/siamcomp/KaleS11},} use the same lazy regularized random walks as in~\cite{goldreich1999a}. As such, it is compatible with our random walks privacy amplification theorem (\Cref{thm:dpamplificationrandwalkintro}), thus implying a private expansion tester.
} {Another problem in this context is \emph{cycle-freeness testing}, where the goal is to distinguish whether the unknown graph is free from having cycles or far from such graphs.} The cycle-freeness testing algorithm~\cite{goldreich2002property} performs a BFS until $8/\dst\maxdeg$ vertices are seen, and can also be described as a $\rd$-disc sampling algorithm. Thus it is compatible with our $k$-disc sampling privacy amplification theorem~\cref{thm:kdiscamplification} by setting $\rd= 8/\dst\maxdeg$. It is an open question whether one could do better than this, taking advantage of the fact that this algorithm does not usually reach its maximum depth.

For future work, it would be very interesting to see if the privacy amplification technique in the dense graph model could be leveraged to design private testers for all graph partition properties (\emph{biclique}, \emph{max-cut}, \emph{max-clique}, \emph{min-bisection} etc.)~\cite{goldreich1998property, DBLP:conf/soda/ShapiraS24}.
{Related to the problem of $H$-freeness testing is the problem of \emph{generalized subgraph freeness testing}~\cite{goldreich2011proximity, DBLP:conf/coco/AdlerK021, DBLP:conf/soda/AdlerKP21}. It would be interesting to see if our {privacy amplification }techniques could be {extended}
to this problem as well.} 
{Moreover, a separation between adaptive testers and non-adaptive testers is known in the property testing literature~\cite{DBLP:journals/siamcomp/GoldreichR11}. 
It would be very interesting to see whether that separation continues to hold under differential privacy as well.} {\emph{The partition oracle} is another important primitive used to design testers in the bounded-degree model~\cite{benjamini2008every, nguyen2008constant,hassidim2009local,newman2011every}. It would be very interesting to design a private variant of partition oracle.} Notably, a private partition oracle could be used to extend our hyperfinite testing result (\Cref{thm:dphyperfinitetestableintro}) in full generality, akin to its non-private counterpart. Finally, {it would be very interesting to see if our techniques could be leveraged for designing private tolerant testers.}

\paragraph{Privacy amplification for graph algorithms beyond property testing. }While our privacy amplification theorem for random walks applies mainly to private testers (or algorithms outputting binary decisions), our privacy amplification theorems for induced subgraph sampling and $k$-disc sampling apply to any edge (and node)-DP graph algorithms that employ the corresponding sampling scheme. 
As such, one can apply our results to graph \emph{estimation} problems, rather than just testing. {While differentially-private sublinear algorithms for estimation problems have been studied previously by~\cite{BlockiGMZ23}, their approach was inherently black-box. It would be interesting to see if our privacy amplification theorems could result in a white-box approach that leads to better accuracy for problems similar to the ones studied by~\cite{BlockiGMZ23}. }%

\paragraph{Private property testing for general graphs. }In our work, we consider the dense and bounded-degree models separately. However, our techniques could apply to the general graph model as well, which encapsulates both the dense and bounded-degree models. 
{While the bipartiteness tester in~\cite{goldreich1999a} only applied to bounded-degree graphs, the follow-up work of~\cite{kaufman2004} demonstrated an algorithm for general graphs combining insights from the dense and bounded-degree bipartiteness testers. It would be very interesting to see if our private bipartiteness testers can be used to design a similar tester in the general graph model.}

\subsection{Additional Related Works}

\paragraph{Differentially Private Hypothesis Testing. }Differential privacy is most often deployed as a means to protect individuals in a database, where each person is modeled as a row with many sensitive attributes. In this setting, one can make the common statistical choice to think of people as samples from a distribution over all people. Under this lens, it is clear why testing distribution properties has been at the forefront of private testing. The problem of privately testing distribution properties has been extensively studied~\cite{DBLP:conf/stoc/CanonneKMSU19, DBLP:conf/icml/CaiDK17, pmlr-v80-aliakbarpour18a, acharya2018differentially, pmlr-v202-kazan23a}.
The core question this area seeks to answer is: \emph{given two distributions, how many i.i.d. samples are needed to distinguish them privately?} For simple hypotheses,~\cite{DBLP:conf/stoc/CanonneKMSU19} characterized a DP variant of the classical Neyman-Pearson lemma. Classical results such as private identity and closeness testing of distributions have also been studied~\cite{pmlr-v80-aliakbarpour18a,acharya2018differentially, DBLP:conf/icml/CaiDK17}. This line of work has since been extended to the high-dimensional setting~\cite{CannoneKMUZ20}, and to private minimax lower-bound techniques~\cite{pmlr-v132-acharya21a}. Similar to our setting, a straightforward approach to designing DP hypothesis testers is applying randomized response to the non-DP hypothesis tester. Under distribution testing, the goal is to reduce the number of (presumably expensive) samples drawn from some distribution. Privacy is then applied to the actually sampled dataset. In contrast, the graph testing model presumes the existence of a massive object which we query at the time of testing. This is closer to the setting of differentially private machine learning, where a large dataset is processed in small random batches, and we aim to avoid paying for privacy for samples we might not see~\cite{abadi2016}.

\paragraph{Differential Privacy in Graphs. } Differential privacy in graphs has grown extensively since its introduction by~\cite{NissimRS07}. We give a few highlights of this large body of work here. The first graph differential privacy model studied was edge-DP. In particular,~\cite{NissimRS07, GuptaLMRT} designed edge-DP algorithms for fundamental problems such as minimum spanning tree cost, triangle counting, minimum cut, and vertex cover size. The stronger notion of node-DP was later studied in~\cite{blocki2013a, kasiviswanathan2013}. Subsequently, DP graph algorithms have been designed for a variety of problems across different computation models. In the continual release model, where the input is a stream of edges over time and the algorithm is required to output at every timestep while preserving DP, both node and edge-DP algorithms have been designed for several graph problems, with edge-DP having received more attention. These include triangle counting, densest subgraph, and maximum matching~\cite{fichtenberger_et_al:LIPIcs.ESA.2021.42, JainSW24, epasto2024power, RaskhodnikovaS24, Song18}. In the local model, edge-DP notions have been more extensively studied for problems such as subgraph counts~\cite{ImolaMC21, EdenLRS25} and densest subgraphs~\cite{DBLP:conf/focs/DhulipalaLRSSY22}. The notion of node-DP in the local model was very recently introduced by~\cite{raskhodnikova2026}.

\paragraph{Organization of the paper.}
The rest of the paper is organized as follows. In \Cref{sec:prelim}, we discuss the necessary preliminaries of this work. Then in \Cref{sec:privatetesterdensegraph}, we present the private testers for the dense graph model, followed by the private testers in the bounded-degree model in \Cref{sec:privatetesterboundeddegreegraph}. %
Proofs of some statements are moved to the appendix.

\section{Preliminaries}\label{sec:prelim}

Throughout this work, we will use the notations $O(\cdot)$ and $\Omega(\cdot)$ to hide the dependencies on parameters $\priv$ and $d$ that we consider to be constants. We will also use $\widetilde{O}(\cdot)$ and $\widetilde{\Omega}(\cdot)$, where we hide poly-logarithmic dependencies on the parameters. $[t]_{+}$ denotes $\max\{t,0\}$. {The notation $\mathbb{P}(Z)$ denotes the set of probability measures on the space $Z$.}
For two discrete distributions $P$ and $Q$ over a common finite support $\mathcal{D}$, the total variation distance between $P$ and $Q$ is defined as $TV(P,Q) = \frac{1}{2} \sum_{x \in \mathcal{D}} |P(x)-Q(x)|$.
For a graph $G=(V,E)$, $V$ and $E$ denote the vertex and edge sets of $G$, respectively. Throughout this work, we will assume the graph $G$ to be undirected and without self-loops or parallel edges. {Also, when we mention sampling, we mean sampling without replacement unless otherwise stated. In particular, this is true of all our results except those for random walks.}

\begin{definition}[$\diverge$-divergence]\label{def:alpha-div}
The $\diverge$-divergence (also known as hockey-stick divergence) with $\diverge\geq1$ between two probability measures $\mu,\mu'\in\mathbb{P}(Z)$ is defined as
\[
D_{\diverge}(\mu\|\mu')=\sup_{E}\bigl(\mu(E)-\diverge\mu'(E)\bigr)
=\int_{Z}\left[\frac{d\mu}{d\mu'}(z)-\diverge\right]_{+}d\mu'(z)
=\sum_{z\in Z}\bigl[\mu(z)-\diverge\mu'(z)\bigr]_{+},
\]
where $E$ ranges over all measurable subsets of $Z$, and the last equality holds for discrete $Z$    
\end{definition}

\paragraph{Differential Privacy.} We define some of the tools and techniques we use from the literature on differential privacy. 

\begin{theorem}[Privacy profiles~\cite{balle2020a}]\label{thm:dp-profile}
A mechanism $\pmech:\dmain\to \mathbb{P}(\range)$ is a mechanism from an input set to distributions over $\range$, equipped with a neighboring relation $\simeq_{\dmain}$. The \emph{privacy profile} of $\pmech$ is the function $\pdelta_\pmech:[0,\infty)\to[0,1]$ given by
\[
\pdelta_\pmech(\priv) \defeq \sup_{x\simeq_{\dmain}x'}D_{e^{\priv}}(\mathcal{M}(x)\|\mathcal{M}(x')).
\]
In particular, this implies that $\pmech$ is $(\priv,\pdelta(\priv))$-DP for all $\priv\geq 0$.
\end{theorem}

Group privacy is a well-studied notion in the context of differential privacy. Intuitively, to protect all groups of size $k$ in a database, one would have to consider databases at distance $k$ from one another. 
\begin{definition}[Group privacy]
\label{def:group-privacy}
Replacing the standard neighboring relation $\simeq$ in~\cref{thm:dp-profile} with $\simeq^k$, its application $k$ times, defines the $k$-\emph{group privacy profile} \[
    \pdelta_{\pmech,k}(\priv) = \sup_{x\simeq^k x'}D_{e^\priv}(\pmech(x)\|\pmech(x')).
\]
The standard analysis of group privacy (see \cite{dwork2013}~for example) yields the following bound
\begin{equation}
    \pdelta_{\pmech,k}(\priv) \leq \frac{e^\priv - 1}{e^{\priv/k} - 1}\pdelta_\pmech(\priv/k)
\end{equation}
\end{definition}

We frequently make use of binary randomized response~\cite{kasiviswanathan2011}, a simple and optimal algorithm for making one bit private.
\begin{definition}[Randomized Response]
\label{def:randomized-response}
    For any $\priv>0$, randomized response $\RR_{\priv}:\{0,1\}\to\{0,1\}$ is defined as \[
    \RR_{\priv}(x) = \begin{cases}
        x &\text{with probability }\frac{e^\priv}{e^\priv + 1},\\
        1-x &\text{with probability } \frac{1}{e^\priv + 1}.
    \end{cases}
    \]
    Randomized response satisfies $\priv$-differential privacy.
\end{definition}

\paragraph{Privacy amplification by subsampling. }
We will also lean heavily on the privacy amplification by subsampling framework as laid out by~\cite{balle2020a}. First and foremost, we will make extensive use of their \emph{advanced joint convexity} theorem.

\begin{theorem}[{Advanced Joint Convexity of $D_{\diverge}$~\cite[Theorem 11]{balle2020a} }]\label{advanced-joint-convexity}
Let $\mu,\mu'\in\mathbb{P}(Z)$ be probability measures satisfying $\mu=(1-\eta)\mu_{0}+\eta\mu_{1}$ and $\mu'=(1-\eta)\mu_{0}+\eta\mu_{1}'$ for \emph{some} $\eta,\,\mu_{0},\,\mu_{1}$ and $\mu_{1}'$. Given $\diverge\geq1$, let $\diverge'=1+\eta(\diverge-1)$ and $\varphi=\diverge'/\diverge$. Then the following holds:
\[
D_{\diverge'}(\mu\|\mu')=\eta D_{\diverge}\bigl(\mu_{1}\|(1-\varphi)\mu_{0}+\varphi\mu_{1}'\bigr).
\]
Expanding this, using regular joint convexity, we get
\[
D_{\diverge'}(\mu\|\mu')\leq\eta\left[(1-\varphi)D_{\diverge}(\mu_{1}\|\mu_{0})+\varphi D_{\diverge}(\mu_{1}\|\mu_{1}')\right].
\]
\end{theorem}

While \cref{advanced-joint-convexity} holds for any decomposition, we use the language of couplings---especially the \emph{maximal coupling}, which exactly characterizes the total variation distance between $\mu$ and $\mu'$ (see~\cref{def:maximal-coupling}).

\begin{definition}[Coupling]\label{def:coupling}
A coupling between two distributions $\nu,\nu'\in\mathbb{P}(Y)$ is a distribution $\pi\in\mathbb{P}(Y\times Y)$ whose \emph{marginals along the projections} $(y,y')\mapsto y$ and $(y,y')\mapsto y'$ are $\nu$ and $\nu'$, respectively.
\end{definition}
\begin{definition}[Maximal Coupling~\cite{balle2020a}]
\label{def:maximal-coupling}
    Suppose $\nu,\nu'$ are distributions with finite support. Let $\nu_{0}(y)=\min\{\nu(y),\nu'(y)\}$ denote the overlap between $\nu$ and $\nu'$, and let $\eta=\text{TV}(\nu,\nu')=1-\sum_{y\in Y}\nu_{0}(y)$, so that $\nu_{0}$ has total mass $1-\eta$, while the residuals $\nu-\nu_{0}$ and $\nu'-\nu_{0}$ each have total mass $\eta$. The maximal coupling is the mixture $\pi=(1-\eta)\pi_{0}+\eta\pi_{1}$, where
\[
\pi_{0}(y,y')=\frac{\nu_{0}(y)\,\mathbf{1}\{y=y'\}}{1-\eta},\qquad
\pi_{1}(y,y')=\frac{(\nu(y)-\nu_{0}(y))(\nu'(y')-\nu_{0}(y'))}{\eta^{2}}.
\]
\end{definition}
\noindent To see the sense in which this coupling is maximal, project $\pi$ along the first marginal:
\begin{equation}\label{eq:proj}
\begin{aligned}
\sum_{y'\in Y}\pi(y,y')
&=\sum_{y'\in Y}\left((1-\eta)\pi_{0}(y,y')+\eta\,\pi_{1}(y,y')\right)\\
&=\nu_{0}(y)+\frac{\nu(y)-\nu_{0}(y)}{\eta}\underbrace{\sum_{y'\in Y}(\nu'(y')-\nu_{0}(y'))}_{=\eta}\\
&=\nu_{0}(y)+(\nu(y)-\nu_{0}(y))=\nu(y),
\end{aligned}
\end{equation}
and symmetrically, the second marginal is $\nu'$, so $\pi$ is indeed a coupling. Moreover $\Pr_{\pi}[y\neq y']=\eta=\text{TV}(\nu,\nu')$, which is the minimum over all couplings, hence \emph{maximal}.

The most useful couplings to a privacy analysis are what we call \emph{distance-bounded couplings}---couplings that demonstrate that everything in the support of one distribution is close to something in the support of the other. This is closely related to the concept of \emph{distance-compatible couplings} used in~\cite{balle2020a}, though we distinguish our definition from theirs, since we do not claim that our couplings are minimal.

\begin{definition}[Distance-bounded coupling]\label{def:dcc} 
Given a coupling $\pi$ between two distributions $\nu,\nu'\in\mathbb{P}(Y)$, we can use that coupling to rewrite the \emph{mixture distributions} $\tilde{\mu}=\nu M$ and $\tilde{\mu}'=\nu'M$ as $\tilde{\mu}=\sum_{y,y'}\pi_{y,y'}\mathcal{M}(y)$ and $\tilde{\mu}'=\sum_{y,y'}\pi_{y,y'}\mathcal{M}(y')$. Using joint convexity, the following holds:
\[
D_{e^{\varepsilon}}(\tilde{\mu}\|\tilde{\mu}')\leq\sum_{y,y'}\pi_{y,y'}D_{e^{\varepsilon}}(\mathcal{M}(y)\|\mathcal{M}(y'))\leq\sum_{y,y'}\pi_{y,y'}\delta_{\mathcal{M},d(y,y')}(\varepsilon),
\]
where $\delta_{\mathcal{M},d(y,y')}(\varepsilon)$ is the \emph{group-privacy profile} at distance $d(y,y')$.
\end{definition}

Following~\cite{balle2020a}, a \emph{subsampling mechanism} formalizes the notion of generating samples from a larger object. In this work, we will only consider subsampling mechanisms that take graphs as input.
In our definition below, we assume that $X$ contains graphs, and $Y$ contains either graphs or tuples of graphs. We distinguish between $X$ and $Y$ since $Y$ will often not be the same set as $X$, such as when we sample $\rd$-discs to produce tuples of rooted graphs. 

\begin{definition}[Subsampling mechanism]
    A \emph{subsampling mechanism} is a randomized algorithm $\sampler:X\to\mathbb{P}(Y)$ that takes as input a graph $G$  and outputs a finitely supported distribution over graphs or lists of graphs.
\end{definition}

We emphasize that while two graphs can be node or edge neighboring---subsamples of those graphs might not inherit the same notions of distance. Thus, throughout the paper, we will define appropriate notions of neighboring samples, as needed. 

{In order to better understand the subsampling mechanisms used in the paper, we give more intuition about the privacy of subsamples and our use of composition next.} Consider a private mechanism $\pmech:Y\to\mathbb{P}(Z)$, and define the subsampled mechanism $\pmech\circ\sampler:X\to\mathbb{P}(Z)$ where $\pmech\circ\sampler(x)=\pmech(\sampler(x))$ and the composition entails feeding a sample of $\sampler(x)$ into $\pmech$. For this to make sense as a privacy analysis, we must assume a neighboring relation $\simeq_Y$ on $Y$ such that $\pmech$ is $(\priv,\pdelta)$-differentially private with respect to $\simeq_Y$. We must also have a neighboring relation $\simeq_X$ on $X$. Then, the question becomes \emph{what are the values of $\priv'$ and $\pdelta'$ such that $\pmech\circ\sampler$ is $(\priv',\pdelta')$-differentially private with respect to $\simeq_X$.}

To help understand the use of composition throughout this work, consider that a non-private tester can be seen as a decision function $f:Y\to\{0,1\}$, composed with a sampling mechanism $\sampler$, so that $\tester=f\circ\sampler$. Similarly, a private mechanism can be seen as a function composed with a noise-channel $Q$ so that $\pmech = Q\circ f$. Under this lens, a subsampled mechanism $\pmech^\sampler = Q \circ f \circ \sampler = \pmech\circ\sampler = Q\circ\tester$. We use each of these decompositions (into tester plus noise, or private mechanism plus subsampling) depending on which aspect of an algorithm we wish to highlight.

\paragraph{Probability amplification of {non-private} testers. }%
To ensure the completeness and generality of our technical exposition, we will use \failp to denote the probability that a tester makes a soundness or completeness error.  If the probability of error in both cases is at most $\failp$, then we will simply say that \tester has probability of error at most \failp.

\begin{remark}
    Note that given a tester with failure probability at most $1/3$, we can produce one with arbitrary failure probability $\failp$ using $\bigO{\log(1/\failp)}$ repetitions of the algorithm. This follows from standard techniques on ``success amplification'' by majority vote. The typical argument goes, given \tester with probability of failure at most $1/3$; the majority, defined as the function $\mathds{1}\left\{\sum_{i=1}^\repeats \tester(x_i) > \repeats/2\right\}$ will equal the correct answer with high probability. Using the Chernoff bound, we can say that the failure probability (that is, the probability that the majority disagrees with the more likely outcome) decreases as $2e^{-\repeats}$, and so setting $\repeats\approx \log(1/\failp)$ yields a new tester with failure probability $\failp$. %
\end{remark}

\paragraph{The dense-graph model.} In the dense graph model, the degree of any vertex is unbounded, and the graph is typically represented by an adjacency matrix. The tester can query any entry of the (unknown) adjacency matrix, commonly known as \emph{edge query}. The notion of closeness and farness is defined as follows.
\begin{definition}[$\dst$-close/far for dense graphs]
\label{def:dense-closeness}
Let $G=([\ns],E)$ be a dense graph with unbounded degree, and let $\dst>0$. We say that $G$ is $\dst$-close to a property $\property$ if one can turn $G$ into a graph satisfying $\property$ by adding or deleting at most $\dst n^2$ edges. Otherwise, $G$ is said to be $\dst$-far from $\property$.
\end{definition}

In this model, many properties can be tested using a number of samples independent of $n$, the number of vertices in the graph. In particular, in the dense-graph model, for any graph property $\property$, if there exists a tester with query complexity $q$, then there exists a \emph{canonical tester} with query complexity $O(q^2)$. The canonical tester works by uniformly subsampling the vertices and checking whether the induced subgraph has the property of interest. This was first shown by \cite{goldreich2003three}. {The canonical tester is described in Algorithm~\ref{algo:canonicaltester}.}

\begin{algorithm}\label{algo:canonicaltester}
\caption{Canonical tester for dense graphs}
\KwIn{Query access to an unknown dense graph $G$, a graph property $\property$, parameter $\dst \in (0,1)$.}
\KwOut{\texttt{Accept} if $G$ satisfies $\property$, and \texttt{Reject} if $G$ is $\dst$-far from $\property$.}

$t \gets s_{\property}$ number of vertices of $G$ (depends on the property $\property$).

$S(\property) \gets$ a uniformly randomly chosen subgraph of $G$ of $t$ vertices.

Query all pairs of vertices in $S(\property)$ to know the induced subgraph $G[S(\property)]$ completely.

Output \texttt{Accept} if and only if $G[S(\property)]$ satisfies the property $\property$.    
\end{algorithm}

\paragraph{The bounded-degree model.} This model, as the name suggests, consists of graphs with bounded maximum degree $d$. In this model, a tester performs \emph{neighbor queries}: given a vertex $v$ and an integer $i \in [d]$, the oracle returns the $i$-th neighbor of $v$, and if $v$ has fewer than $i$ neighbors, the oracle returns a special symbol $\bot$. 
Similar to the dense graph model, the notion of closeness and farness is defined as follows. 

\begin{definition}[$\dst$-close/far for bounded-degree graphs]
\label{def:bounded-closeness}
Let $G=([\ns],E)$ be a graph with maximum degree $d$, and let $\dst \in (0,1]$. We say that $G$ is $\dst$-close to a property $\property$ if one can turn $G$ into a graph satisfying $\property$ by adding or deleting at most $\dst \cdot dn$ edges. Otherwise, $G$ is said to be $\dst$-far from $\property$.
\end{definition}

{
}

Now let us define the notion of hyperfinite graphs, which will be useful for our results in the bounded-degree model.

\begin{definition}[Hyperfinite graph and property]
Let $\dst \in (0,1]$ and $k \in \mathbb{N}$ be some parameters. A graph $G$ is said to be $(\dst, k)$-hyperfinite if one can remove at most $\dst d n$ edges from $G$ such that the remaining graph has connected components of size at most $k$ vertices. For a function $\rho: \mathbb{R}^{+} \rightarrow \mathbb{R}^{+}$, $G$ is $\rho$-hyperfinite if $G$ is $(\dst, \rho(\dst))$-hyperfinite. A property (aka. a collection of graphs) is hyperfinite if every graph in the property is $\rho$-hyperfinite for some function $\rho$. %
\end{definition}

\section{Private Testers for Dense Graphs}\label{sec:privatetesterdensegraph}

In this section, we present our results in the dense graph model. We will first present our edge-DP privacy amplification result in \Cref{thm:induced-edge-amplification}. Then, using this amplification result, we design the DP canonical tester in \Cref{thm:dpcanonicaltester}. Finally, we design a private $H$-freeness tester in \Cref{thm:private-POT-h-freeness}.

To present our edge privacy amplification result, we first need to formalize the concepts of induced subgraph sampling and edge distance.
Let us start with induced subgraph sampling.

\begin{definition}[Induced subgraph sampling]
\label{induced-subgraph-sample}
Let $\ns \geq \ms h$ for some positive integers $\ns, \ms, h$ and $G=([\ns],E)$ be an unknown graph. An induced subgraph sampler $\subsampler_{\ms,h}$ accesses $G$ via edge queries and outputs a distribution over $\ms$-tuples of $h$-vertex graphs. To obtain a sample from $\subsampler_{\ms,h}$, the following steps are performed:
\begin{itemize}
    \item[(i)] Draw $(v_{1},\dots,v_{\ms h})$ uniformly at random without replacement from $[\ns]$. Let $B = (b_1,\dots,b_\ms)$ where for every $i \in [\ms]$, $b_i = (v_{(i-1)h + 1},\dots, v_{ih})$ and $b_i(j) \defeq v_{(i-1)h + j}$ is its $j$'th entry.

    \item[(ii)] Output $S = (s_1,\dots, s_\ms)$ where for every $i \in [\ms]$, $s_i = ([h],E_i)$ and $E_i = \{ \{j,j'\}\subseteq [h]: \{b_i(j), b_i(j')\}\in E \}$.
\end{itemize}
For $\ms=1$, we will abbreviate $\subsampler_{\ms,h}$ as $\cansampler\defeq\subsampler_{1,h}$.

\end{definition}

Several testers in the dense graph model first sample one or more fixed-sized subgraphs from the unknown graph and then proceed to analyze the sampled subgraphs. The above definition of induced subgraph sampling captures this sampling phase of such testers. Notably, this approach is used for canonical testers, which we will discuss later in detail in this section.

Now to analyze and use the induced subgraph sampling based private testers, we need to define a neighboring relation between two samples (two $\ms$-tuples of $h$-vertex graphs) obtained from the induced subgraph sampler $\subsampler_{\ms,h}$. Note that the following notion of distance is sensitive to the order of samples.

\begin{definition}[Edge distance and edge-neighboring samples]
\label{def:edge-edit-distance}
Let $\range_{\ms,h}$ denote the set of $\ms$-tuples $y=(y_1,\dots,y_\ms)$ of graphs that can be obtained as the output of the induced subgraph sampler $\subsampler_{\ms,h}$, where for every $i \in [\ms]$, $y_i=([h],E_i)$.\footnote{Every $y_i \in y$ is a $h$-vertex graph with common vertex set $[h]$, they only differ in their corresponding edge sets.} Consider two samples $y,y'\in\range_{\ms,h}$ with $y \neq y'$. The \emph{edge distance} between $y$ and $y'$ is equal to the sum of the size of the \emph{symmetric difference}, denoted $\triangle$, between their edge sets:
\[
d(y,y') \defeq \sum_{i=1}^{\ms} |E_i \triangle E_i'|.
\]
We say these two samples $y$ and $y'$ are edge neighboring (denoted as $y\simeq_e y'$) if $d(y,y')=1$. 
\end{definition}

As mentioned in the introduction, very recently \cite{dwork2026efficientprivatepropertytesting} defined the notion of node-DP amplification. To better motivate our edge-DP amplification result, we first mention the node-DP amplification result from \cite{dwork2026efficientprivatepropertytesting}.

\begin{lemma}[{Node-privacy amplification for induced subgraph sampling~\cite{dwork2026efficientprivatepropertytesting}}]
\label{lem:node-dp-induced-subgraphs}
Let $\pmech$ be a graph algorithm which satisfies $(\priv, \pdelta)$\emph{-node differential privacy}. Consider the algorithm $\pmech'=\pmech\circ\subsampler_{\ms,h}$, which is designed by composing $\pmech$ with the induced subgraph sampler $\subsampler_{\ms,h}$. Then $\pmech'$ satisfies $(\priv',\pdelta')$\emph{-node differential privacy}, where $\priv' = \log\left( 1 + \frac{\ms h}{\ns}(e^{\priv} - 1)\right)$ and $\pdelta' = \frac{\ms h}{\ns}\pdelta$.
\end{lemma}

While our result for edge-DP amplification by subsampling looks similar to the node-DP amplification result stated above, it is far from immediate. On one hand, it has an improved dependence on $\ns$, going from around $1/\ns$ to $1/\ns^2$. This makes sense, since to query an edge we must sample both of its endpoints. On the other hand, the dependence on $\priv$ is worsened, incurring an additional $2(h-1)$ group-privacy profile term in our coupling argument. The $2(h-1)$ term is the maximum edge-distance between a sample that contains $e$ and a sample that doesn't under our coupling.

\begin{restatable}[Edge-privacy amplification for induced subgraph sampling]{theorem}{inducedsubgraph}\label{thm:induced-edge-amplification}
Let $\pmech:\range_{\ms,h}\to\mathcal{Z}$, where $\mathcal{Z}$ is any output domain, be an algorithm satisfying $(\priv, \pdelta)$\emph{-differential privacy} with respect to the edge-neighboring relation $\simeq_e$~(\cref{def:edge-edit-distance}). Then the process of sampling $y\sim\sampler_{\ms,h}(G)$ and computing $\pmech(y)$ satisfies $(\priv', \pdelta')$-edge differential privacy where \[
\priv' = \log \left( 1 + \frac{\ms h(h-1)}{n(n-1)}(e^{2(h-1)\priv} - 1) \right) ,
\]
and \[
\pdelta' = \frac{\ms h(h-1)}{n(n-1)} \frac{e^{2(h-1)\priv} - 1}{e^{\priv} - 1}\pdelta.
\]
\end{restatable}

We give a detailed, but incomplete sketch of the proof here and leave the probability accounting details to~\Cref{sec:subgraph-amplification-proof}.

\begin{proof}[Proof sketch.]
Let $G\simeq_e G'$ be neighboring graphs differing on the edge $e=\{x_1,x_2\}$, $\omega = \subsampler_{\ms,h}(G)$ and $\omega'=\subsampler_{\ms,h}(G')$ be the corresponding distributions over $\range_{\ms,h}$.

Recall that the sampler draws $(v_1,\dots,v_{\ms h})$ uniformly without replacement from $[\ns]$, before dividing them into the $\ms$ ordered blocks $b_i\defeq(v_{(i-1)h+1},\dots,v_{ih})$, whose $j$'th entry we denote $b_i(j)$. Let $\blockdist\defeq \blockdist(\ms,h,\ns)$ be the distribution over block configurations $B=(b_1,\dots,b_\ms)$ defined by this process. We abuse notation and write $u\in b_i$ (respectively $\{u,v\}\subseteq b_i$) to mean that $u$ is an entry of the tuple $b_i$ (respectively that both $u$ and $v$ are). 

To move from block configurations to subgraphs, we define two mappings $g$ and $g'$ that fill in the edges between vertices. Since these are deterministic functions, it will suffice to reason about $\blockdist$ in our coupling. Formally, let $g$ be the function that takes as input a block configuration $B=(b_1,\dots,b_\ms)$ and outputs $y\defeq g(B)\defeq (y_1,\dots,y_\ms)$, where $y_i\defeq([h],E_i)$ and $E_i\defeq\{\{j,j'\}\subseteq[h]:\{b_i(j),b_i(j')\}\in E(G)\}$, exactly as in the definition of the sampler $\subsampler_{\ms,h}$~(\cref{induced-subgraph-sample}). Similarly, define $g'$ with respect to $E(G')$. Then, if $\kappa_g$, $\kappa_{g'}$ are the {Markov} kernels induced by $g$ and $g'$ respectively, then $\ome = \blockdist \kappa_g$ and $\ome' = \blockdist \kappa_{g'}$.

We will need to bound the probability that $e$ is seen while sampling induced subgraphs according to~\cref{induced-subgraph-sample}. 
Given $B=(b_1,\dots,b_\ms)\sim\blockdist$, let
$$
S := \left\{\exists i\in[\ms]\text{ such that }\{x_1,x_2\} \subseteq b_i\right\}.
$$
In other words, \(S\) is the event that the two endpoints of the differing edge $e=(x_1,x_2)$ are sampled in the same block. Then
$$\eta:= \Pr[S].$$ 
Note that $\eta$ is an upper bound on the total variation distance between $\omega$ and $\omega'$. Since there are $\binom{\ns}{2}$ possible pairs of vertices, and the probability that we see one of them is equal to the probability that they are selected among the first $\ms$ blocks, which together contain $\ms \binom{h}{2}$ pairs, we get that \[
\eta ={\frac{\ms \binom{h}{2}}{\binom{n}{2}}} = \frac{\ms h (h-1)}{\ns(\ns - 1)}.
\]

\noindent Using the notation of~\cite{steinke2022}, let $\blockdist_0\defeq \blockdist|_{S^c}$ denote $\blockdist$ conditioned on $S^c$. Similarly, define $\blockdist_1\defeq \blockdist|_{S}$. Let, 
\begin{equation*}
    \omega_1    \defeq \omega|_{S} = \blockdist_1\kappa_g,\qquad
    \omega_1'   \defeq \omega'|_{S} = \blockdist_1\kappa_{g'},\qquad
    \omega_0    \defeq \omega|_{S^c} = \blockdist_0\kappa_g
\end{equation*}
Note that we could equivalently write $\omega_0 = \omega'|_{S^C}$, since $\blockdist_0\kappa_g = \blockdist_0\kappa_{g'}$. Then, for every $y \in \range_{\ms,h}$ {(which is the set of $\ms$-tuples of graphs on the vertex set $[h]$)}, the law of total probability gives us, 
\begin{align*}
    \omega(y) &= \Pr[S^c]\Pr[g(B)=y\mid S^c] + \Pr[S]\Pr[g(B)=y\mid S]  \\
    &= (1-\eta)\cdot\omega_0(y)+\eta\cdot\omega_1(y)
\end{align*}
Since the identity holds pointwise, and we can derive a similar identity for $\omega'(y)$, we can write them distributionally as 
\begin{align*}
    \omega = (1-\eta) \omega_0 + \eta \omega_1,\qquad
    \omega' = (1-\eta) \omega_0 + \eta \omega_{1}'.
\end{align*}

Using the decomposition described above, we can apply advanced joint convexity~(\cref{advanced-joint-convexity}). 
Let $\pmech$ be an 
$(\priv,\pdelta)$-differentially private algorithm with respect to $\simeq_e$ (\cref{def:edge-edit-distance}), and let $M$ be the Markov kernel defined by \pmech. Let $\mu=\omega M$ and $\mu'=\omega' M$, and setting $\mu_0\defeq\ome_0 M$, $\mu_1\defeq\ome_1 M$ and $\mu_1'\defeq\ome_1' M$, so that $\mu = (1-\eta)\mu_0 + \eta\mu_1$ and $\mu' = (1-\eta)\mu_0 + \eta\mu_1'$, we get
\begin{align}
    D_{\diverge'}(\mu\|\mu') & = \eta D_{\diverge}(\mu_1 \| (1-\varphi)\mu_0 + \varphi\mu_1')\nonumber\\
    &\leq \eta\left[ (1-\varphi)D_{\diverge}(\mu_1\|\mu_0) + \varphi D_{{\diverge}}(\mu_1\|\mu_1') \right].\label{eq:joint-convexity-expansion}
\end{align}
The second inequality is by joint convexity. Our task is then to bound the $D_{\diverge}$ terms using couplings and group privacy profiles. Specifically, we will construct distance-bounded couplings between the pairs $\mu_1$ and $\mu_1'$ and $\mu_1$ and $\mu_0$. 

As a sketch, consider the problem of constructing the coupling between $\mu_1$ and $\mu_0$. Recall that $\mu$ is defined by the application of our private mechanism to a sample from $\ome$. Taking one step further in the chain of abstractions, consider the distributions $\blockdist_1$ and $\blockdist_0$, that are just over the vertex set. For the sake of this sketch, let $B\sim\blockdist_1$ be a sample from the distribution of vertex sub-samples conditioned on $\{x_1,x_2\}$ appearing in the same block (which would of course lead to $e$ being included in one of the subgraphs). A coupling can be described as a mapping from samples, such as $B$ from $\blockdist_1$ to samples that are distributed according to $\blockdist_0$. We characterize samples from $\blockdist_0$ as being in one of three possible states. (1) Both $x_1$ and $x_2$ get sampled, but they fall into different blocks. To transform $B$ into a sample of this form, we can move one of $x_1$ or $x_2$ into another block. (2) One of $x_1$ or $x_2$ is sampled, but the other is not. To transform $B$ into a sample of this form, we can move one of the endpoints out of the sample, replacing it with a new random element. (3) Neither $x_1$ or $x_2$ are sampled. In this case we simply have to remove $x_1$ and $x_2$ from $B$, replacing both with uniformly random vertices that weren't otherwise included in the sample.

An algorithm that turns samples like $B$ into samples from $\lambda_0$ with the right probabilities can be thought of as a coupling. For the full details of this coupling, and the simpler one between $\mu_1$ and $\mu_1'$, we refer to~\Cref{sec:subgraph-amplification-proof}, where the full probabilities and distance calculations are worked out. Here it suffices to say that our coupling yields bounded edge distance at most $2(h-1)$. Applying group-privacy at that distance yields the final theorem.
\end{proof}
\Cref{thm:induced-edge-amplification} is {extremely} general in that it applies to any graph algorithm that samples subgraphs in the way described. 
This generality however does not come for free, and for decision problems, we can in fact prove the following stronger statement, which drops the $2(h-1)$ group privacy parameter in the exponent of $\priv'$. 

\begin{theorem}[Improved privacy amplification for testers]
\label{thm:rr-subgraph-tester}
Consider an arbitrary Boolean function $f:\range_{\ms,h}\to \{0,1\}$.
Let $\pmech$ be $\priv$-randomized response applied to the output of $f$, that is, $\pmech = \RR_{\priv}\circ f$. Then $\pmech\circ\subsampler_{\ms,h}$ is $\priv'$-edge-DP with $\priv' = \log(1 + \eta(e^{\priv} - 1))$, $\eta=\frac{\ms h(h-1)}{\ns(\ns-1)}$. 
\end{theorem}
\begin{proof}
Let $G\simeq_e G'$ be edge-neighboring graphs differing on the edge $e=\{x_1,x_2\}$, $\omega = \subsampler_{\ms,h}(G)$ and $\omega'=\subsampler_{\ms,h}(G')$ be the corresponding distributions over $\range_{\ms,h}$. Recall from the proof sketch above that $S$ denotes the event that both endpoints of $e$ are sampled in the same block, that
\[
\eta = \Pr[S] = \frac{\ms h(h-1)}{\ns(\ns-1)}
\]
and that conditioning on $S$ yields the decomposition
\begin{align*}
    \omega = (1-\eta) \omega_0 + \eta \omega_1,\qquad
    \omega' = (1-\eta) \omega_0 + \eta \omega_{1}',
\end{align*}
where $\omega_0 = \omega|_{S^c} = \omega'|_{S^c}$, $\omega_1 = \omega|_{S}$ and $\omega'_1 = \omega'|_{S}$.

Let $M$ be the Markov kernel defined by $\pmech=\RR_{\priv} \circ f$, and set 
\[
\mu = \omega M,\quad \
\mu' = \omega' M,\quad\
\mu_0 = \omega_0 M,\quad\
\mu_1 =\omega_1 M,\quad\
\mu'_1 = \omega'_1 M,
\]
so that $\mu = (1-\eta) \mu_0 + \eta \mu_1$ and $\mu' = (1-\eta) \mu_0 + \eta \mu_1'$. Let $\diverge= e^{\priv}$, $\diverge' = 1 + \eta (\diverge - 1)=e^{\priv'}$ and $\varphi= \diverge'/\diverge$, so that $\priv' = \log(1+ \eta (\eep - 1))$ as in the theorem statement. Applying advanced joint convexity~(\cref{advanced-joint-convexity}) to these decompositions, following from joint convexity,
\begin{equation}
    \label{eq:rr-advanced-joint-complexity}
    D_{\diverge'}(\mu\|\mu') = \eta D_{\diverge}\left( \mu_1\| (1-\varphi)\mu_0 + \varphi\mu'_1 \right) \leq \eta \left[ (1-\varphi)D_{\diverge}(\mu_1\|\mu_0) + \varphi D_{\diverge}(\mu_1\|\mu'_1) \right].
\end{equation}

It remains to bound the two divergences on the right-hand side of~\cref{eq:rr-advanced-joint-complexity}. In contrast to the more general proof, we do not need distance-bounded couplings for this step. As such, let $\pi$ be an arbitrary coupling between $\omega_1$ and $\omega_0$. By~\cref{def:dcc}, 
\[
D_{\diverge}(\mu_1\|\mu_0) \leq \sum_{y,y'} \pi(y,y') D_\diverge(\pmech(y)\|\pmech(y')) \leq \sup_{a,b\in\{0,1\}} D_{\diverge}(\RR_{\priv}(a)\|\RR_{\priv}(b)),
\]
where the last inequality holds because $\pmech(y)=\RR_{\priv}(f(y))$ and $f$ takes values in $\{0,1\}$. The same argument repeated for $\omega_1$ and $\omega'_1$ yields an identical bound on $D_\diverge(\mu_1\|\mu'_1)$.

Using the definitions of randomized response~(\cref{def:randomized-response}) and pure differential privacy in terms of $\diverge$-divergences (\Cref{def:alpha-div}) to write\[
D_\diverge(\mu_1\|\mu_0) \leq \left[ \frac{\eep}{\eep + 1   } - \diverge \frac{1}{\eep + 1}\right]_+
\]
where this is clearly 0 whenever $\diverge\geq \eep$. Therefore, the right-hand side of~\cref{eq:rr-advanced-joint-complexity} is clearly 0. We can bound $D_{\diverge'}(\mu'\|\mu)$ by repeating this argument with $G$ and $G'$ swapped. Therefore the subsampled mechanism is $\priv'$-edge differentially private with $\priv' = \log(1 + \eta (\eep - 1))$. 
\end{proof}

It is important to note that this tightening is entirely unnecessary for node-DP, where any private mechanism will yield the same result.

Now we are ready to discuss the canonical tester. Recall that canonical testers are a subclass of testers in the dense graph model that allows lifting any adaptive tester to a non-adaptive one, with only quadratic blow-up in its query complexity. The result is formally stated as follows.

\begin{theorem}[{\cite[Theorem 4.5]{goldreich2003three}}]
\label{def:canonical-tester}
    Let $\prop$ be any graph property and $\tester$ be an arbitrary (possibly adaptive) $\dst$-tester for $\prop$ with completeness and soundness error at most $1/3$.
    Suppose $\tester$ has query complexity $q(\ns,\dst,1/3)$,\footnote{Recall that the expression $q(\ns,\dst,1/3)$ denotes the query complexity of $\tester$ when the unknown graph has $n$ vertices, the proximity parameter is $\dst$ and the failure probability of $\tester$ is at most $1/3$.} then there exists a non-adaptive canonical tester for $\prop$ that samples a random subgraph on $2q(\ns,\dst,1/6)$ vertices\footnote{Note that in the worst case the canonical tester incurs an additional factor of 2 in the error probability, hence the new sample complexity depends on a failure probability of $1/6$. Additionally, the query complexity of this algorithm is $O(2q^2)$, since we have to query all combinations of edges in the worst case. We will distinguish between \emph{sampling} vertices and \emph{querying} edges wherever it is ambiguous.} and tests $\prop$ with completeness and soundness error at most $1/3$. If $\tester$ has one-sided error, then the canonical tester also has one-sided error.
\end{theorem}

The canonical tester as stated above is clearly compatible with~\cref{thm:rr-subgraph-tester}. Putting the two together demonstrates the existence of private analogues of all canonical testers with query complexity $q=o(\sqrt{\ns})$. Clearly this includes all constant-query property testers and many more with non-constant query complexities.

\begin{theorem}[Edge-private canonical tester] \label{thm:dpcanonicaltester} 
Let $G=([\ns],E)$ be an unknown graph on $\ns$ vertices, \tester be a non-private canonical property tester for some property \prop, $\dist, \failp\in(0,1)$ and $\priv>0$ are the proximity, failure probability, and privacy parameters, respectively. \tester samples $q\defeq q(\ns,\dst,\failp/2)$ vertices from $G$ uniformly at random {without replacement}, queries the induced subgraph on the sampled vertices using at most {$q^2$ queries}, and rejects any graph at least $\dst$-far from \prop with probability at least $1-\frac{\failp}{2}$. Then there exists a private tester $\tester'$ that is the application of $\RR_{\priv}$ to the output of \tester, with $\priv = \log(1 + \frac{ \ns(\ns-1)}{q(q - 1)}(\eepp - 1))$
. $\tester'$ is an $\priv'$-\emph{edge differentially private}, canonical \emph{two-sided error} property tester with the same query complexity as \tester and failure probability \[
    \failp'\leq \frac{\failp}{2} + \frac{q(q-1)}{\ns(\ns-1)(e^{\priv'} - 1)}=\bigO{\max\left\{\failp, \frac{q^2}{\ns^2(e^{\priv'} - 1)}\right\}}.
\]
\noindent In particular, if $\priv' \geq \log(1 + \frac{4q^2}{\ns^2\failp})$, then $\tester'$ is an $\priv$-edge DP {$\dst$}-tester with query complexity $q^2$ and failure probability $\failp$.
\end{theorem}
\begin{proof}
    First, note that~\cref{thm:rr-subgraph-tester}
    implies that this process is private as claimed. The probability that $\RR_{\priv}$ flips any bit it receives as input is $1/(\eep + 1)$ (see~\cref{def:randomized-response}), therefore the probability that this process incurs an error is
    \begin{align*}
        \Pr[\RR_{\priv}\circ \tester\text{ errors}] &= \Pr[\tester\text{ errors}]\cdot\Pr[\RR_{\priv}\text{ does not flip input}] + \Pr[\tester\text{ correct}]\cdot \Pr[\RR_{\priv}\text{ flips input}]\\
        &\leq \Pr[\tester\text{ errors}] + \Pr[\RR_{\priv}\text{ flips input}]\\
        &\leq \failp' + \frac{1}{\eep + 1 } \\%\leq \failp + e^{-\priv_0}\\
        &\leq\failp' +  \frac{q(q-1)}{\ns(\ns-1)(\eepp - 1)}.
    \end{align*}
    where the last inequality is by~\cref{thm:rr-subgraph-tester} with $\ms=1$ and $h=q$. Setting $\failp'=\failp/2$, we can see that the total failure probability is at most $\failp$ when $\priv'\geq \log(1+\frac{2q(q-1)}{\ns(\ns-1)\failp})$.
\end{proof}

For completeness, we here state the equivalent result for node-DP, which follows immediately along the same lines, substituting where appropriate the edge-DP amplification statement for \cref{lem:node-dp-induced-subgraphs}.
{
\begin{corollary}[Node-private canonical tester]
\label{cor:node-dp-canonical}
    Let $G=([\ns],E)$ be an unknown graph on $\ns$ vertices, $\tester$ be a non-private canonical property tester for some property \property, $\dst,\failp\in(0,1)$ and $\priv>0$ are the proximity, failure probability and privacy parameters, respectively. \tester samples $q\defeq q(\ns,\dst,\failp/2)$ vertices from $G$ uniformly at random without replacement, queries the induced subgraph using at most $q^2$ queries, and rejects any graph that is $\dst$-far from \property with probability at least $1-\frac{\failp}{2}$. Then there exists a private tester $\tester'$ that is the application of $\RR_\priv$ to the output of \tester, with $\priv=\log(1 + \frac{\ns}{q}(\eepp - 1))$. $\tester'$ is an $\priv'$\emph{-node differentially private}, canonical \emph{two-sided error} property tester with the same query complexity as \tester and failure probability\[
    \failp' \leq \frac{\failp}{2} + \frac{q}{\ns(\eepp - 1)} = \bigO{\max \left\{ \failp, \frac{q}{\ns(\eepp - 1)} \right\} }.
    \]
    In particular, if $\priv\geq \log(1 + 4\frac{q}{\ns\failp})$, then $\tester'$ is an $\priv$-node DP \dst-tester with query complexity $q^2$ and failure probability $\failp$.
\end{corollary}
}
There are several interesting observations. Firstly, any constant-query non-private testable properties in the dense graph model can be privately tested via this canonical tester, incurring almost no additional error.  
{
\begin{remark}[On almost having one-sided error]
\label{rem:almost-one-sided}
While the private canonical tester has two-sided error, the error will not be wholly symmetric. {If we slightly expanded our proof} to distinguish between completeness and soundness error, then we would see that for one-sided error testers the private probability of failure will be $O(q^2/(\ns^2 \priv))$ on the completeness side, while the soundness error is $O(\failp + q^2/(\ns^2 \priv))$. As such, one might want to consider these algorithms as being ``almost one-sided.''
\end{remark}
}

To gain a better understanding of how the additional error in the private canonical tester affects the query complexity, we demonstrate its effect on bipartiteness testing in the dense graph model. It is known from \cite{goldreich1998property} that the canonical bipartiteness tester in this model samples $\widetilde{O}(1/\dst^2)$ vertices from the unknown graph uniformly at random and queries the induced subgraph of the sampled vertices. The tester outputs accept if and only if the sampled induced subgraph is bipartite. This requires $\widetilde{O}(1/\dst^4)$ queries for success probability at least $2/3$.

\begin{corollary}[Private canonical bipartiteness tester for dense graphs]
    \label{coro:biptestdense}
    For all $\priv=\bigOmetilde{\log(1+\frac{1}{\dst^4\ns^2})}$, there exists a {$\priv$-edge} DP $\dst$-tester for bipartiteness testing. If the query complexity of the non-private tester is $q=\bigOtilde{1/\dst^4}$, then the private tester has query complexity at most $4q$.
\end{corollary}
\begin{proof}
    The non-private canonical tester for bipartiteness samples $\ms=\bigOtilde{1/\dst^2}$ vertices and inspects the induced subgraph for bipartiteness testing and decides based on the sampled subgraph. Let us assume that this tester has the standard failure probability of at most $1/3$. Repeating this algorithm twice gives a new algorithm with failure probability at most $1/9$ (since the original tester has one-sided error). Applying~\cref{thm:dpcanonicaltester} to this algorithm, we see that for all $\priv\geq \log(1 + \frac{2(2\ms)(2\ms - 1)}{2\ns(\ns-1)(2/9)})$, the private tester has failure probability at most 1/3.
\end{proof}

One should rightly note that there are bipartiteness testers with improved query complexity in the dense graph model~\cite{DBLP:journals/siamdm/AlonK02}; however, they are not canonical. For better exposition, we only present the private \emph{canonical} bipartiteness tester in this work. Secondly, while dense-graph bipartiteness testing can be done with query complexity independent of the size of the graph, for the bounded-degree graph model, a more sophisticated bipartiteness tester needs to be designed, which we discuss later in \Cref{sec:random-walks}.

\subsection{Testing Subgraph Freeness}
We now move to subgraph freeness testing in dense graphs. We first define the subgraph freeness property and the corresponding testing problem, which we will refer to as $H$-freeness testing when $H$ is the subgraph of interest.

\begin{definition}[Subgraph freeness~\cite{goldreich2017introduction}]\label{defi:subgraphfreness}
Let $H$ be a fixed graph. A graph $G=(V,E)$ is $H$-free if $G$ contains no subgraph that is isomorphic to $H$.
\end{definition}

\begin{definition}[$H$-freeness testing]
Let $G=([\ns], E)$ be an unknown graph, $H$ be a known, {connected},\footnote{Subgraph freeness testing when $H$ is not a connected graph is certainly possible non-privately~\cite{goldreich2017introduction}, but we do not consider it in this work.} graph and $\dst \in (0,1]$ be a parameter. Given query access to $G$, the goal of \emph{$H$-freeness testing} is to distinguish with probability at least $2/3$ whether $G$ does not contain $H$ as a subgraph, or is $\dst$-\emph{far}-from $H$-free.
\end{definition}

{
In the statement of our privacy amplification theorem (\cref{thm:induced-edge-amplification}), the quantity $\ms\cdot h(h-1)$, where $\ms$ and $h$ are the number and size of the sampled subgraphs respectively, is clearly connected to the number of \emph{edges} sampled in the procedure. There are many practical instances where, unlike the canonical tester, we might sample many small subgraphs and would rather not pay quadratically for edge queries between those subgraphs that we do not need. A perfect example of this kind of tester is the $H$-freeness tester.
}

If a graph is $\dst$-far from being $H$-free, then an estimable fraction of the graph must be made of copies of $H$. This fact, along with the query complexity of known $H$-freeness testers, is captured by the \emph{graph removal lemma}.

{%
\citet[\S 8.4, Corollary 8.19]{goldreich2016two} showed that $H$-freeness testing and proximity oblivious testers are directly connected via \emph{Szemeredi's regularity lemma}~\cite{szemeredi1975regular}. For our purposes, we will use the following graph removal lemma, which follows from Szemeredi's regularity lemma.
}
\begin{lemma}[Graph Removal Lemma~{{\cite[Theorem 1.1]{ConlonF13}}}]
\label{lem:graph-removal}
Let $G$ be a dense graph with $\ns$ vertices. For every fixed graph $H$ and every $\dst \in (0,1]$, there exists $\prox = \prox(\dst, H) > 0$ such that if $G$ contains fewer than $\prox \cdot \ns^{v(H)}$ copies of $H$, then $G$ can be made $H$-free by removing at most $\dst \ns^2$ edges, {where $v(H)$ denotes the number of vertices in $H$}.
\end{lemma}

Since our goal is to test whether an unknown graph is far from being $H$-free, we will make use of the contrapositive of \Cref{lem:graph-removal}. If a graph is $\dst$-far from being $H$-free, then it must contain \emph{at least} $\prox\cdot \ns^{v(H)}$ copies of $H$. Therefore, a random set of $v(H)$ vertices will contain a copy of $H$ for a $\prox$ fraction of the time. Doing this once is the ``proximity-oblivious tester'' of \citet{goldreich2016two}. Clearly, repeating this test $O(1/\prox)$ times yields a constant probability of witnessing the subgraph $H$ in a far from $H$-free graph. The following is implicit in~\cite[\S 8.4.8-19]{goldreich2017introduction}

\begin{corollary}[Testing $H$-freeness, {\cite[\S 8.4.8-19]{goldreich2017introduction}}]
Let $G=([\ns],E)$ be an unknown dense graph, $H$ be a known graph, and $\dst \in (0,1]$ be a {proximity} parameter. If $G$ is $\dst$-far from $H$-free, then there exists a one-sided error $H$-freeness tester with query complexity $O(v(H)^2/\prox(\dst,H))$, where $\prox(\dst,H)\cdot n^{v(H)}$ is the number of copies of $H$ in $G$.
\end{corollary}

The above corollary follows from the graph removal lemma. Since $G$ is $\dst$-far from $H$-free, following \cref{lem:graph-removal}, there are at least $\prox(\dst,H)\cdot n^{v(H)}$ copies of $H$ in $G$. Therefore, a random selection of $v(H)$ vertices from $G$ will witness a copy of $H$ with probability at least $\prox$. Thus, repeating this process $O(1/\prox)$ times will yield a one-sided error $H$-freeness tester with query complexity $O(v(H)^2/\prox(\dst,H))$. The non-private algorithm (\Cref{alg:Hfreenessdensegraph}) is described in~\cref{sec:algorithms}. %

In fact, the above method of testing is closely related to the concept of \emph{proximity-oblivious testers}~\cite{goldreich2017introduction}. Proximity Oblivious testers are a set of testers that repeat some \emph{basic test} a number of times depending on the proximity parameter $\dst$, while the basic test is independent and oblivious to the proximity parameter. This model was introduced in the seminal work of \cite{goldreich2011proximity}, where the authors studied one-sided error testers. Later, \cite{goldreich2016two} also studied this class of testers in the two-sided error setting.

{As we mentioned in the introduction, under node differential privacy, there is absolutely no distinction between this tester and the canonical tester. This is because in both cases we sample effectively the same number of vertices, and under node-DP the number of vertices is the parameter that controls privacy amplification. In the case of edge differential privacy, the circumstances are different and we may not wish to pay for edge queries we never use. This is particularly salient for $H$-freeness, taking triangle freeness as an example: testing whether a graph is triangle-free while using the proximity oblivious tester relies on sampling an enormous quantity (see~\cref{rem:tower-type}) of constant-sized subgraphs.}

\begin{theorem}[Private $H$-freeness testing in dense graph model]\label{thm:private-POT-h-freeness}
Let $G=([\ns],E)$ be an unknown dense graph, $H$ be a known graph with $h=v(H)\ll \ns$ vertices, $\dst \in (0,1]$ be a parameter and $\prox(\dst,H)$ be as in~\cref{lem:graph-removal}. There exists an $\priv'$-edge differentially private, two-sided error property tester for $H$-freeness with query complexity\[
\bigO{\frac{h^2\log(1/\failp)}{\prox(\dst, H)}}
\]
that succeeds with probability $1-\failp$ so long as $\priv'=\bigOme{\log(1 + \frac{h^2\log(1/\failp)}{\ns^2\failp\prox})}$. 
\end{theorem}
\noindent {Notably, the query complexity of the private $H$-freeness tester is within a constant factor of the non-private query complexity under reasonable assumptions on $\priv'$.}

\begin{proof}
   First, note that since the basic non-private test succeeds with probability $\prox$ and has one-sided error, we can repeat it $\ms\defeq {\log(1/\failp')/\prox}$ times in order to succeed with probability $1-\failp'$. Each test performs at most $2h^2$ edge queries. Then, applying~\cref{thm:rr-subgraph-tester}, solving for $\priv$ as a function of $\priv'$, we see that setting
   $\priv=\log(1 + \frac{\ns(\ns-1)}{\ms h (h-1) }(e^{\priv'} - 1))$ makes the algorithm private. This process errs with probability 
   \begin{align*}
   \failp'' &\leq \failp' + \frac{h(h-1)\log(1/\failp')}{\ns(\ns-1) (e^{\priv'} - 1)\prox}\\
    &\leq \failp' + \frac{h^2\log(1/\failp')}{\ns^2 \priv'\prox}.
   \end{align*}
   Setting $\failp' = \failp/2$, we have $\failp'' \leq \failp$ whenever $\priv'\geq \log(1 + \frac{2h(h-1)\log(2/\failp)}{\ns(\ns-1)\failp \prox})=\bigOme{\log(1+\frac{h^2\log(1/\failp)}{\ns^2\failp\prox})}$. {The query complexity bound follows from combining the query complexities in each run of the tester.}
\end{proof}

{Notice that the canonical tester (\Cref{thm:dpcanonicaltester}) for this problem would have query complexity at least $\bigO{h^2/\prox^2}$ and would pay just as much in the application of the privacy amplification lemma. For a fixed, constant-sized subgraph $H$, this can make an enormous difference, the following remark (\Cref{rem:tower-type}) makes this clear.}

Despite the simplicity of the graph removal lemma as we stated it, the regularity lemma is a rather heavy piece of machinery. As such, the best known bounds on $\prox$ for general $H$ are of the ``tower form''. They are independent of $\ns$, but enormous functions of $\dst$. Fortunately, for bipartite $H$ much better bounds are known, and improvements to these bounds immediately imply improvements to testing algorithms that rely on them. We summarize them in the following remark.

\begin{remark}[Dependence of $\prox$ on $\dst$]\label{rem:tower-type}
The magnitude of $\prox(\dst, H)$ depends critically on the structure of $H$:
\begin{itemize}
    \item If $H$ is \textbf{bipartite}, then $\prox(\dst, H) = \poly(\dst)$, yielding
    polynomial sample complexity in $1/\dst$ \cite{alon2002}. %
    \item If $H$ is \textbf{non-bipartite}, the best known bounds are of tower-type in
    $1/\dst$~\cite{GishbolinerS25}.\footnote{That is to say $\prox$ is of the form $1/T(m)$ where $T(m)=\exp(T(m-1))$ and $T(1)=2$~\cite{goldreich2017introduction}.} 
\end{itemize}
\end{remark}

\section{Private Testers for Bounded-Degree Graphs}\label{sec:privatetesterboundeddegreegraph}
In this section, we present our private graph property testers for bounded-degree graphs. We begin by presenting an amplification technique for regularized random walks in \Cref{thm:dpamplificationrandwalk} and use it to design a bipartiteness tester for bounded-degree graphs in \Cref{thm:dpbipartitebounded}. Then we define the notion of $k$-disc sampling, and prove a privacy amplification result for $k$-disc sampling in \Cref{thm:kdiscamplification}. Finally, as applications of the privacy amplification technique, we first design a private $H$-freeness tester in the bounded-degree model in \Cref{thm:dp-h-freeness}, and then show that every property of hyperfinite graphs is privately testable in \Cref{thm:dphyperfinitetestable}. 

\subsection{Random walk based private testers}
\label{sec:random-walks}
A fundamental sampling technique for sublinear graph algorithms is the random walk. Random walks are fundamental to many testing results in the bounded-degree model. %
We first define regularized random walks\footnote{These are also referred to as ``canonical random walks''~\cite{goldreich2011expansion}; or just as ``random walks'', depending on the paper. We make the distinction clear since it is crucial to our privacy analysis.} and the corresponding sampling scheme in terms of transcripts. Note that where we refer to sampling in this section, we mean sampling \emph{with} replacement.

\begin{definition}[Regularized random walk]
\label{def:reg-random-walk}
    A regularized random walk is a form of lazy random walk where at each step the algorithm either remains at the current vertex $v$ with probability $1 - \frac{d(v)}{2\maxdeg}$, or progresses to a neighboring vertex $w$ of $v$ with probability $1/2\maxdeg$, where $d(v)$ denotes the degree of the vertex $v$. This walk has the property that when started at a uniformly random vertex $v$, the next vertex it visits will also be uniformly random.
\end{definition}

Now let us formally define random walk sampling.

\begin{definition}[Random walk sampling]
\label{def:random-walk-transcript}
    Let $G=([\ns],E)$ be a $\maxdeg$-bounded graph, $\ms,\repeats,\length\geq1$.
    We can sample a \emph{transcript} \transcript from the space of regularized random walks $\transcript=(\transcript_1,\dots,\transcript_{\ms\cdot\repeats\cdot\length})\sim\spacetrans_{\ms,\repeats,\length}(G)$, by choosing $\ms$ random vertices uniformly at random and starting $\repeats$ random walks~(\cref{def:reg-random-walk}) of length $\length$ from each of them. We use $\transcript_i$ to denote the $i$'th vertex in the transcript. For example $\transcript_{\length + 1}$ is the start of the second walk, $\transcript_{\length\cdot \repeats+1}$ is the second starting-point, and $\transcript_{2}$ is the first random step (which may be no step at all, if the walk remains stationary). 
    Note that when we refer to the length of a walk, we mean the total number of steps taken, including steps where the walk remains at the same vertex.
\end{definition}

Now we are ready to present our privacy amplification result for regularized random walks.

\begin{theorem}[Privacy amplification for {binary decisions on} regularized random walks]\label{thm:dpamplificationrandwalk}
Let $f:\spacetrans\to\{0,1\}$ be an arbitrary Boolean function from the space of random walk transcripts~(\cref{def:random-walk-transcript}) to binary decision. Let $\pmech=\RR_{\priv}\circ f$ be the function that, given as input a transcript of regularized random walks, applies $f$ and applies randomized response with privacy parameter $\priv$ to the output of $f$. If those random walks are sampled as in~\cref{def:reg-random-walk} with  $\ms\cdot\repeats\cdot\length\leq\tsize$, then the mechanism satisfies $\priv'$-differential privacy where \[
\priv' = \log \left( 1 + \frac{2T}{\ns}(e^{\priv} - 1) \right).
\]    
\end{theorem}

\begin{proof}
    Let $\transcript\sim\spacetrans_{\ms,\repeats,\length}$ be a transcript as defined in~\cref{def:random-walk-transcript}, and let $f:\spacetrans\to\{0,1\}$ be a (possibly randomized) function from transcripts to binary decisions. Let $\tsize = \ms\cdot \repeats \cdot \length$ be the total number of vertices seen in the course of the random walks, counting repetitions.

    Consider two bounded-degree graphs $G$ and $G'$ which differ in one edge $e$. Transcripts on $G$ and $G'$ will be identical unless one of the random walks happens to visit one or both endpoints of $e$. Let $H$ be the event that either endpoint of a particular edge $e=(v,w)$ is visited. {We define the event $H$ with respect to }either endpoint, since the presence of $e$ adjacent to $v$ and $w$ changes the probability of ``waiting'' on that vertex.

    Fix a vertex $v$. By the fact that the stationary distribution of a regularized random walk is uniform, the marginal distribution of $\transcript_i$ is uniform for every step $i$, meaning $\Pr[\transcript_i=v]=1/\ns$. 
    Let $X_i^{v}$ be the event that $\transcript_i = v$, $X^v$ be the event that $v$ is visited at any point in $\transcript$ that is to say $X^v=\bigcup_{i=1}^\tsize X_i^{v}$. Then by a union bound over the $X_i^v$'s
    \[
    \Pr[X^v] \leq \sum_{i=1}^\tsize \Pr[X_i^{v}] = \frac{\tsize}{\ns}.
    \]
    Since the event $H=X^v\cup X^w$, we can use a second union bound over the two endpoints to see that\[
    \Pr[H] \leq \Pr[v\in\transcript] + \Pr[w\in\transcript] \leq \frac{2\tsize}{\ns}.
    \]

    Formally, let $\omega,\omega'$ be the distribution of transcripts generated from $\spacetrans_{\ms,\repeats,\length}(G)$ and $\spacetrans_{\ms,\repeats,\length}(G')$ respectively. Use the fact that $\operatorname{TV}(\ome,\ome')\leq\Pr[H]\leq \frac{2\tsize}{\ns}$ and set $\eta^*=\frac{2\tsize}{\ns}$ as an upper bound on $\eta$.
    The maximal coupling implies the following decomposition,
    \begin{align*}
        \omega & =(1-\eta)\omega_{0} + \eta \omega_{1} \\
        \omega' & = (1-\eta)\omega_{0}+ \eta \omega_{1}',
    \end{align*}
    where $\eta\leq\eta^{*}$ since $\eta=\operatorname{TV}(\omega,\omega')\leq\Pr[H]\leq \eta^{*}$.

Now, let $\mu=\omega M$, $\mu'=\omega'M$ be the application of $f$ and $\text{RR}_{\priv}$ to the transcripts sampled from $\omega,\omega'$. We can apply advanced joint convexity~\cref{advanced-joint-convexity}
    . Set $\diverge=e^{\priv}$, let $\diverge'=1+\eta(\diverge-1)$, and $\diverge^* = 1 + \eta^* (\diverge - 1)$, $\beta=\diverge'/\diverge$. Note that $\diverge'\leq \diverge^*$ since $\eta\leq \eta^*$. Then
    \begin{align*}
        D_{\diverge^*}(\mu\|\mu')
        &\le D_{\diverge'}(\mu\|\mu')
        \tag{$\diverge\mapsto D_\diverge$ non-increasing}\\
        &= \eta D_{\diverge}\big(\mu_1 \| (1-\beta)\mu_0+\beta\mu_1'\big)
        \tag{advanced joint convexity}\\
        &\le \eta\big[(1-\beta)D_{\diverge}(\mu_1\|\mu_0)+\beta D_{\diverge}(\mu_1\|\mu_1')\big]
        \tag{joint convexity}\\
        &\le \eta \sup_{a,b\in\{0,1\}} D_{e^{\priv}}\!\big(\text{RR}_{\priv}(a) \| \text{RR}_{\priv}(b)\big)
        =0.
    \end{align*}
    
    \noindent The last inequality is implied by the following: {Since binary randomized response is $\priv$-DP over its entire input space, for every $a,b\in\{0,1\}$ and every output event $E$, 
    \[
    \Pr[\mathrm{RR}_{\priv}(a)\in E]\le e^{\priv} \Pr[\mathrm{RR}_{\priv}(b)\in E].
    \] 
    By the definition of pure differential privacy in terms of $\diverge$-divergence  
    \[
    D_{e^{\priv}} \bigl(\mathrm{RR}_{\priv}(a)\|\mathrm{RR}_{\priv}(b)\bigr)=0.
    \]
}
Therefore,
\begin{align*}
    \sup_{a,b\in \{ 0,1 \}}D_{e^{\priv}}(\text{RR}_{\priv}(a)\|\text{RR}_{\priv}(b)) =0.
\end{align*}
Hence the subsampled mechanism is $\priv'$-DP with $\priv'=\log\!\big(1+\eta^{*}(e^{\priv}-1)\big)$.
\end{proof}

Due to its fundamental importance, we design a private bipartiteness tester in this model.
We first state the celebrated bipartiteness tester of~\cite{goldreich1999a}. The non-private tester (\Cref{alg:biptest}) is presented in \Cref{sec:algorithms}.

\begin{theorem}[Bipartiteness tester~\cite{goldreich1999a,goldreich2002property}]
Let $G$ be an unknown bounded-degree graph with $n$ vertices, maximum degree $d$ and $\gamma \in (0,1]$ be a {proximity} parameter. Given query access to $G$, there exists an algorithm that can distinguish whether $G$ is bipartite or $\gamma$-far from bipartite by performing $\tilde{O}_{\gamma}(\sqrt{n})$ queries. Moreover, this bound is tight, up to polylogarithmic factors.    
\end{theorem}

In~\cite{goldreich1999a}, the random walk used is a \emph{lazy regularized random walk} (as in \Cref{def:reg-random-walk}), where at each step the algorithm either remains at the current vertex with probability $1-\frac{d(v)}{2\maxdeg}$, or progresses to vertex $w$ with probability $\frac{1}{2\maxdeg}$, where $d(v)$ denotes the degree of the vertex $v$. This walk has the property that---being doubly stochastic---when started at a uniformly random vertex $w_0$, the next vertex it visits, $w_1$ will be uniformly random. 
As such, each step of the walk looks locally uniformly random. 

{In contrast, the elegant recent work of \cite{fei2026} shows that bipartiteness testing can be achieved using walks of length $\bigO{\log\ns}$ (a dramatic improvement on the prior known length of $\bigO{\log^6\ns}$) and uses \emph{simple} {lazy random walks} in their analysis.} The simple lazy random walk chooses an adjacent edge to traverse uniformly at random, and remains stationary at each step with constant probability $1/2$. Equivalently, it is a simple random walk with a random binomial length. They also do not start their walks from a uniformly random vertex, rather they sample a starting location proportional to the degree of that vertex. Formally, the stationary distribution of this random walk is the degree-weighted distribution $p(v)=d(v)/(2|E|)$, while the stationary distribution of the regularized random walk is the uniform distribution. It seems much more challenging to show any pure-DP privacy amplification for their algorithm. Consider the graph on $n$ vertices with two edges and a neighboring graph which has one edge. We will see said edge with constant probability whereas a walk started from a random vertex will see said edge with probability $O(1/\ns)$. Thus, achieving any privacy amplification for bipartiteness testing in logarithmic rounds seems difficult without changing the underlying nature of their algorithm or analysis.

In our work, we demonstrate both dense and bounded-degree bipartiteness testing algorithms, yet it is open whether our results imply a bipartiteness tester in the general graph model.

\begin{theorem}[Privacy of the bipartiteness tester]\label{thm:dpbipartitebounded}
For all $\priv\geq \log(1 + \frac{12\tsize}{\ns})$, there exists a $\priv$-differentially private \dst-tester for bipartiteness in bounded degree graphs with query complexity $\tsize=\bigO{\sqrt{\ns} \cdot \poly(\log (\ns)/\dst)}$.%
\end{theorem}

\begin{proof}
    The bipartiteness tester by \citeauthor{goldreich1999a}~\cite{goldreich1999a} first samples $\ms=\poly(1/\dst)$ random vertices, and then from each sampled vertex, starts $\repeats=\bigO{\sqrt{\ns}}$ regularized random walks, each of length $\length=\poly(\log(\ns))$, rejecting if there exists an odd-cycle among the vertices explored. Let $f$ be the function on transcripts that outputs 1 if it witnesses an odd cycle, or 0 otherwise. Let $\tsize = \ms\cdot\repeats\cdot\length$ as above. Let $\pmech'$ be the tester that samples $\transcript\sim\spacetrans_{\ms,\repeats,\length}(G)$, computes $f$ and applies $\RR_{\priv_0}$ to the output. Applying~\cref{thm:dpamplificationrandwalk} with \[
    \priv_0 = \log\left( 1 + \frac{n}{2\tsize}(e^\priv - 1) \right).
    \] shows that $\pmech$ satisfies $\priv$-differential privacy. 
    
    If $f$ applied to the randomly sampled \transcript has non-private failure probability $1/6$, then we can compute that \pmech has failure probability \[
        \failp \leq \frac{1}{6} + \frac{2\tsize}{\ns(e^\priv - 1)}.
    \]
    Where $\failp\leq 1/3$ whenever $\priv \geq \log(1 + \frac{12T}{\ns})$, and $\failp\leq \frac{4T}{\ns\priv}$ otherwise. The final tester has two-sided error (where the original has one-sided error).%
\end{proof}

\subsection{\texorpdfstring{$k$}{k}-disc sampling based private testers}

In this section, we present our results on $k$-disc sampling based private testers. A \emph{$k$-disc} around a vertex $v$ is the subgraph induced by all vertices of distance at most $k$ from $v$. We formalize this notion and the notion of $k$-disc sampling below. 

\begin{definition}[$k$-discs and $k$-disc samples~\cite{newman2013every}]\label{def:k-disc}
A connected graph $G=(V,E)$ with a specially marked vertex $v$ as the root is defined to be a \emph{rooted graph}. A rooted graph $G$ with a root vertex $v$ has radius $\rad(G)=k$ if every vertex in $V$ is at distance at most $k$ from $v$. For a $d$-bounded degree graph $G$, an integer $k$ and a vertex $v$, let $B_{G}(v,k)$ be the subgraph rooted at a vertex $v$ that is induced by all vertices of $G$ that are at distance at most $k$ from $v$. $B_{G}(v,k)$ is called the \emph{$k$-disc} around $v$. A $k$-disc sample from a graph $G$ is obtained by first sampling a vertex $v$ uniformly at random from $V$ and then returning the $ k$-disc $B_{G}(v,k)$ around $v$.
\end{definition}

Now let us define the notion of $k$-disc sampler.

\begin{definition}[$k$-disc sampler]
    Let $\kdsampler$ be an algorithm that takes a graph $G=(V,E)$ as input and outputs a distribution on $\ms$-tuples of $\rd$-discs, each with a distinct root in $G$.  To generate a sample $S=(s_1,\dots,s_\ms)\sim\kdsampler(G)$, draw  $v_1,\dots,v_\ms$ {uniformly} without replacement from $V$, perform a breadth-first search of depth $\rd$ around each sampled vertex, and return a list of $v_i$-rooted graphs, $s_i$ for every $i \in [\ms]$.
\end{definition}

We use the standard notion of Hamming distance between our tuples.
\begin{definition}[$\rd$-disc tuple distance]
\label{def:k-disc-distance}
    Let $\kdrange$ denote the set of $\ms$-tuples of rooted graphs of radius at most $\rd$. For $y,y'\in\kdrange$, define the following notion of distance\[
    d(y,y') \defeq \sum_{i=1}^\ms \indic{y_i \neq y_i'},
    \]
    where $y=y'$ is the standard notion of equality between graphs---that they have the same vertex and edge sets without relabeling. We say that $y$ and $y'$ are \emph{neighbors} $y\simeq y'$ if and only if $d(y,y')=1$.
\end{definition}

Using our definitions for $\rd$-disc sampling, and our neighboring relation, we can state our privacy amplification result. 
\begin{theorem}[Privacy amplification for $k$-disc sampling]\label{thm:kdiscamplification}
    Let $\pmech:\kdrange\to\mathcal{Z}$ where $\mathcal{Z}$ is any output domain, be an algorithm that satisfies $(\priv,\pdelta)$-differential privacy with respect to the $\rd$-disc tuple distance. Let $\pmech'=\pmech\circ \kdsampler(G)$ be the algorithm that applies $\pmech$ to a $\rd$-disc sample $y\sim \kdsampler(G)$. $\pmech'$ satisfies $(\priv',\pdelta')$-\emph{edge} and \emph{node-differential privacy} where
    \[
   \priv' = K \log\left(1 + \frac{\ms}{\ns}(\eep - 1)\right),\qquad \pdelta' = K\frac{\ms}{\ns}\left(1 +  \frac{\ms}{\ns}(\eep -1)  \right)^K \pdelta,
    \]
    and $K\defeq 2\maxdeg^{rd} + 1$. 
\end{theorem}
\begin{proof}
For brevity, the following proof is written with respect to node-DP, but we note in-text where the proof can be modified to yield results for edge-DP.

Let $G=([\ns],E)$ and $G'=([\ns],E')$ be graphs of maximum degree at most $\maxdeg$ such that $G\simeq_v G'$ are vertex-neighboring; that is to say that they share a common vertex set, but differ in all of the edges adjacent to one vertex $v$. Notice that if we produced the list of all distinct $\rd$-discs (having distinct roots) from $G$ and $G'$, they will be identical except for the $\rd$-discs containing $v$. Upper bounding this quantity, we see that the number of $\rd$-discs that can contain $v$ is at most $K\defeq 2\maxdeg^\rd + 1$ which is also an upper bound on the number of vertices in one $\rd$-disc. Our proof therefore proceeds as follows (1) Consider the list of all $\rd$-discs and apply privacy amplification by subsampling without replacement, and then (2) apply group privacy for groups of size $K$.

Recall that $\pmech$ is an algorithm that satisfies $(\priv,\pdelta)$-differential privacy with respect to the $\rd$-disc tuple distance (see \Cref{def:k-disc-distance}). Applying privacy amplification by subsampling without replacement~\cite{balle2020a}, we see that subsampling $\ms$ discs from this ``database'' will yield a new protocol $\pmech'$ that satisfies $(\priv_1,\pdelta_1)$-differential privacy where\[
\priv_1 = \log\left(1 + \frac{\ms}{\ns}(\eep  - 1) \right),\qquad \pdelta_1 = \frac{\ms}{\ns}\pdelta.
\]
Therefore, this algorithm also satisfies $(\priv', \pdelta')$-group privacy for groups of size $K$ (see~\cite{balle2020a}), where \[
\priv' = K \priv_1,\qquad \pdelta' = Ke^{K\priv_1}\pdelta_1.
\]
Since we already established that group-privacy at distance $K$ over tuples of distinct sampled $\rd$-discs is sufficient to guarantee node-DP for the underlying graph, this is enough to complete the proof. Expanding $\priv'$ and $\pdelta'$, we see that if $\pmech$ is an algorithm that satisfies $(\priv,\pdelta)$-differential privacy with respect to the $\rd$-disc tuple distance then $\pmech'$ will satisfy $(\priv',\pdelta')$-node differential privacy where\[
\priv' = K \log\left(1 + \frac{\ms}{\ns}(\eep - 1)\right),\qquad \pdelta' = K\frac{\ms}{\ns}\left(1 +  \frac{\ms}{\ns}(\eep-1)  \right)^K \pdelta.
\]
Notice that if we chose instead to upper bound the number of $\rd$-discs that could contain some \emph{edge}, we would still find that $K\leq2\maxdeg^{\rd} + 1$. As such, the same proof holds for edge-differential privacy.
\end{proof}

In the bounded-degree model, the problem of $H$-freeness testing is again one of the most natural problems to study since it aligns so closely with how we sample from our graphs. 
\begin{lemma}[{Detection probability of a $\rd$-disc sample~\cite[\S 9.2]{goldreich2017introduction}}]
    Let $H$ be a graph with $h=v(H)$ vertices. If $G=([\ns],E)$ is $\dst$-far from being $H$-free in the bounded-degree model, then a $\rd$-disc sample witnesses a copy of $H$ with probability at least $\dst/2$.
\end{lemma}

The algorithm associated with the above lemma is presented in \Cref{alg:pot-h-free}.
Repeating this test $O(\log(1/\failp)/\dst)$ times yields~\Cref{alg:gen-h-free}, a general algorithm for $H$-freeness with one-sided error probability at most $\failp$.
Combining, we have the theorem below.

\begin{theorem}[Non-private $H$-freeness testing in the bounded-degree model \cite{goldreich2002property}]
Let $G$ be an unknown graph with maximum degree $d$, $H$ be a known graph with $h$ vertices, $\dst \in (0,1], \failp \in (0,1)$ 
be the proximity parameter and failure probability, respectively. There exists a constant query tester (\Cref{alg:gen-h-free}) (independent of $n$ and depends only on $\dst$ and $\failp$) that with probability at least $1-\failp$, distinguishes whether $G$ is $H$-free or $\dst$-far from $H$-free. 
\end{theorem}

With the above theorem in hand, we apply randomized response~(\cref{def:randomized-response}) with $\priv = \log(1 + \frac{\ns}{\ms}(e^{ \priv'/K } - 1))$ to the output of~\Cref{alg:gen-h-free}, we get the following theorem about private $H$-freeness tester in the bounded-degree model.

\begin{theorem}[Private $H$-freeness testing in bounded-degree model]
\label{thm:dp-h-freeness}
Let $H$ be a graph with $h$ vertices and radius $\rd$. For all $\priv\geq \frac{K\log(2/\failp)}{\dst\failp\ns}$, there exists an $\priv$\emph{-differentially private}, two-sided error property tester for $H$-freeness with failure probability $\failp$. The algorithm samples $q=\frac{\log(2/\failp)}{\dst}$ vertices, and makes up to ${q\cdot 2\maxdeg^\rd+1}$ queries while performing a BFS around each sampled vertex. Note that this is within a constant factor of the non-private query complexity. The privacy guarantees hold for both edge and node differential privacy. 
\end{theorem}

\begin{proof}
    Let $q=\log(1/\failp')/\dst$, $K=2\maxdeg^\rd+1$. Sample the $q$-tuple $y\sim\sampler^{\text{disc}}_{q,\rd}(G)$ and let $f$ be the function that outputs $1$ if $y$ contains a copy of $H$. By~\cref{thm:kdiscamplification}, applying randomized response to the output of $f$ with $\priv=\log\left(1 + \frac{\ns}{q}(e^{\priv'/K} -1 )\right)$ will satisfy $\priv'$-differential privacy. 
    
    Calculating the new failure probability, we see that
    \begin{align*}
        \beta''& \leq \Pr[\text{non-private tester errors}] + \Pr[\text{randomized response flips output}]\\
        &= \failp' + \frac{1}{e^\priv + 1}\\
        &\leq \failp' + \frac{\log(1/\failp')}{\dst\ns(e^{\priv'/K}-1)}\\
        &\leq \failp' + \frac{K\log(1/\failp')}{\dst\failp\ns}.
    \end{align*}
    Finally, set $\failp'= \failp/2$ and solve for $\priv$ to see that the private tester has failure probability at most \failp whenever $\priv'\geq \frac{2K \log(2/\failp)}{\dst\failp\ns}$. Otherwise the tester has the failure probability of $\failp'' \leq \frac{2K\log(1/\failp')}{\dst\priv'\ns}$.
\end{proof}

\begin{theorem}[{Every hyperfinite property is testable~\cite{newman2013every}}]
\label{thm:non-priv-hyperfinite}
Let \property be a $\rho$-hyperfinite family of graphs with maximum degree $\maxdeg$, and let $\dst,\failp\in(0,1]$. Given query access to an unknown bounded-degree graph $G$, there exists an algorithm (\Cref{alg:hyperfinite-tester}) that performs constant number of queries (independent of $n$, and depends only on $\dst$ and $\failp$) and with probability at least $1- \failp$, distinguishes whether $G \in \property$, or $\dst$-far from \property.
{\cref{alg:hyperfinite-tester} is a (two-sided error) property tester for \property}. %
\end{theorem}

Hyperfinite graphs are a generalisation of planar graphs including, for example, every class of graphs defined by not containing a specific forbidden minor. In a seminal work,~\citet{newman2013every}, demonstrated that every property of hyperfinite graphs is (constant-query) testable. They showed that it is sufficient to just explore a number (independent of $\ns$) of randomly sampled $\rd$-discs of the graph. Applying~\cref{thm:kdiscamplification}, we get the following private version of their result. 

{The proof of \Cref{thm:non-priv-hyperfinite} has two parts. In the first part, one checks whether the unknown input graph is hyperfinite, and then, if yes, in the second step, checks for the property \property. The first step is achieved by the use of a technique known as \emph{partition oracle}. However, as mentioned in the introduction, we currently do not have a private variant of the partition oracle, and as a result, our private hyperfinite tester works under the additional assumption that the unknown input graph is promised to be hyperfinite. We note that our resulting tester is private (as long as the input graph has bounded degree), and only the accuracy is affected in the event that the hyperfinite promise is violated. }

\begin{theorem}[Every property of hyperfinite graphs is \emph{privately} testable]
\label{thm:dphyperfinitetestable}
    Let $\rho,\dst,\failp\in(0,1)$, and let \property be an arbitrary property of $\rho$-hyperfinite graphs of maximum degree \maxdeg. Let $\eta=\eta(\dst,\rho,\maxdeg)$ and $\rd=\rd(\dst,\rho,\maxdeg)$ be the parameters described in~\cref{alg:hyperfinite-tester}~\cite{newman2011every}. Let $N=N(\rd,\maxdeg)$ be the number of non-isomorphic rooted graphs of radius at most $\rd$, and set \[
    s\defeq \lceil \eta^{-2}N^2\log(40N/\failp)\rceil,\qquad K\defeq 2\maxdeg^\rd+1.
    \]
    Then for every $\priv>0$ there exists an $\priv$-node and edge differentially private, two sided error $\dst$-tester for \property on inputs $G=([\ns],E)$ with the following properties
    \begin{itemize}
        \item If $G$ is hyperfinite, then the tester has failure probability at most $\failp$ so long as $\priv=\bigOme{\frac{s\maxdeg^\rd}{\failp\ns}}$.
        \item If $G$ is bounded-degree but not hyperfinite, then there are no accuracy guarantees, but the tester is still private.
    \end{itemize}
    The tester samples $s$ vertices uniformly at random and performs a BFS of depth $\rd$ around each, making $\bigO{s\cdot \maxdeg^\rd}$ queries in total. Notably, the query complexity is independent of $\ns$ and within a constant factor of the non-private query complexity of~\cref{thm:non-priv-hyperfinite}.
\end{theorem}

\begin{proof}
    For the privacy analysis, we use a restricted version of~\cref{alg:hyperfinite-tester}, only performing the second phase of the algorithm. Applying~\Cref{thm:kdiscamplification} with $s$ as above, $\rd=\rd(\dst,\rho,\maxdeg)$ as above, and $K=\maxdeg^\rd+1$. Let $\pmech$ be the algorithm that given a list of $\rd$-discs decides as in the second phase~\cref{alg:hyperfinite-tester} and then applies randomized response with $\priv=\log(1 + \frac{\ns}{s}(e^{\priv'/K} - 1))$. Then $\pmech$ run on randomly sampled $\rd$-discs from $G$ satisfies $\priv'$-differential privacy and incurs additional failure probability $\bigO{\frac{s\cdot K}{\ns \priv'}}$. Repeating the test before applying randomized response yields the dependence on $\failp$ as in the theorem statement.
\end{proof}

\begin{remark}
We note that \Cref{alg:hyperfinite-tester} follows the proof structure of \cite{newman2011every}, the conference version of \cite{newman2013every}. As noted by the authors in \cite{newman2013every}, they revised the structure in the later version for better clarity. However, the algorithm directly obtained from the revised structure uses the partitioning oracle in phase~2 of the algorithm, which would significantly complicate the analysis of our private tester later on.    
\end{remark}

\paragraph{AI Disclosure.} 
The authors used a variety of commercially available LLMs, including but not limited to Claude, Gemini, and ChatGPT. These services were primarily used as ``advanced spell-checkers''----assisting with writing, grammar, and punctuation. We here also include a specific list of instances where these tools were used for more involved writing tasks, or to develop the ideas or mathematics of the paper:
\begin{itemize}
    \item The proof of~\cref{thm:induced-edge-amplification} was assisted by asking Claude Opus 4.8 to construct a coupling between two distributions over permutations where the indices are divided into blocks---an abstraction of the coupling used in our paper. A human interpretation of the final version of coupling, produced by several rounds of back and forth discussion, plus some heavy reinterpretation by the authors, is presented in the paper.
    \item The \LaTeX ~for \Cref{alg:hyperfinite-tester} was generated by Gemini given as input the paper~\cite{newman2011every}.
\end{itemize}
The authors verified the correctness and originality of all content including references.

\printbibliography
\newpage
\appendix
\section{Proof of edge-privacy amplification by induced subgraph sampling}
\label{sec:subgraph-amplification-proof}

In this section, we complete the proof of \Cref{thm:induced-edge-amplification}.

\inducedsubgraph*

\begin{proofof}{\cref{thm:induced-edge-amplification}}
We continue from where the proof sketch left off.

Following the methodology of~\cite{balle2020a}, let $\pmech$ be an $(\priv_0,\pdelta_0)$-differentially private protocol with respect to the edge distance relation $\simeq_{e}$~(Definition~\ref{def:edge-edit-distance}), and let $M$ be the Markov kernel defined by said protocol. This allows us to rewrite the distribution of outputs of our protocol as $\mu=\omega M$, where $\ome=\subsampler_{\ms,h}(G)$ and $\ome'=\subsampler_{\ms,h}(G')$ are our sampling distributions and $M$ is the Markov kernel defined by $\pmech$. In this way, $\mu=\ome M$ and $\mu'=\ome' M$.

Our task is then to bound the $D_{\diverge}$ terms using couplings and group privacy profiles. Specifically, we will construct distance bounded couplings between the pairs $\mu_1$ and $\mu_1'$ and $\mu_1$ and $\mu_0$.

The first case, between $\mu_1$ and $\mu_1'$ is fairly straightforward. Define the coupling $\pi_{1,1}$ between $\ome_1$ and $\ome_1'$ defined as follows
\begin{enumerate}
    \item Sample $B\sim\blockdist_1$;
    \item Construct $(y,y') = (g(B), g'(B))$.
\end{enumerate}
Since $\ome_1,\ome_1'=\blockdist_1\kappa_{g},\blockdist_1\kappa_{g'}$, this is clearly a valid coupling. Furthermore, for all pairs $y,y'$ generated in this way $d(y,y')\leq 1$. Applying~\cref{def:dcc}, we see that for every $\priv>0$,
\begin{equation}
    D_{e^{\priv}}(\mu_1\|\mu_1') \leq \sum_{y,y'}\pi_{1,1}(y,y')D_{e^{\priv}}(\pmech(y)\|\pmech(y')) \leq \pdelta_\pmech (\priv).\label{eq:distance-bound-11}
\end{equation}
{The distance $d(y,y')$ is the edge distance on $\range_{\ms,h}$ of~\cref{def:edge-edit-distance}. Since $y=g(B)$ and $y'=g'(B)$, they are tuples of graphs that differ on just the edge $(x_1, x_2)$ in one block. As such, $d(y,y')\leq 1$.}

The second case {between $\mu_1$ and $\mu_0$} is more complicated. We first define in more detail the sample space over our block configurations. Let
\[
\mathcal{B}:=\left\{B=(b_1,\ldots,b_\ms):b_i\in[\ns]^h,\text{ and the }\ms h\text{ entries of }B\text{ are pairwise distinct}\right\}.
\]
Since the sampler draws its $\ms h$ vertices uniformly without replacement, $\blockdist$ is uniform on $\mathcal{B}$, and hence the conditional distributions defined below are uniform on the corresponding subsets of $\mathcal{B}$. Let $\mathbf B$ denote the random block configuration produced by the subsampling procedure. Recall that $e=(x_1,x_2)$ is the edge on which $G$ and $G'$ differ. We partition $\mathcal B$ according to how the endpoints $x_1,x_2$ appear in the sampled blocks. Define
\begin{align*}
    \mathcal{B}_S&:=\left\{B=(b_1,\ldots,b_\ms)\in\mathcal B:\exists r\in[\ms]\text{ such that }\{x_1,x_2\}\subseteq b_r\right\},\\
    \intertext{and}
    \mathcal{B}_T&:=\left\{B\in\mathcal B:x_1,x_2\text{ are both sampled, but occur in different blocks}\right\},\\
    \mathcal{B}_H&:=\left\{B\in\mathcal B:\text{exactly one of }x_1,x_2\text{ is sampled}\right\},\\
    \mathcal{B}_M&:=\left\{B\in\mathcal B:\text{neither }x_1\text{ nor }x_2\text{ is sampled}\right\}.
\end{align*}
These sets form the disjoint partition
\[
\mathcal{B}\setminus\mathcal{B}_S
=\mathcal{B}_T\sqcup \mathcal{B}_H \sqcup \mathcal{B}_M.
\]
The probability of each case is as follows: $T$ occurs when we sample both $x_1$ and $x_2$, but they appear in different blocks. There are $\binom{\ms h}{2}$ pairs they can occur in, and we have to remove the $\ms\binom{h}{2}$ pairs that occur within one block, normalizing we see that
\[
p_T =\frac{ \binom{\ms h}{2} - \ms\binom{h}{2} }{\binom{n}{2}} = \frac{h^2 \ms (\ms - 1)}{n(n - 1)}.
\]
$H$ occurs when we sample exactly one endpoint. Let $d\defeq \ns - \ms h$ denote the number of remaining elements after sampling. Multiplying these probabilities and doubling them to account for the two symmetric cases shows that\[
p_H = \frac{2d\ms h}{n(n-1)}.
\]
Straightforwardly, $M$ just requires that neither end is sampled, which occurs with probability\[
p_M = \frac{d(d-1)}{n(n-1)}.
\]

With this, we can restate $\lambda_1$ as the distribution of block configurations
conditioned on $\mathbf B\in\mathcal B_S$, and let $\lambda_0$ denote
the distribution conditioned on
$\mathbf B\notin\mathcal B_S$:
\[\lambda_1(B):=\Pr[\mathbf B=B\mid \mathbf B\in\mathcal B_S],\qquad
\lambda_0(B):=\Pr[\mathbf B=B\mid \mathbf B\notin\mathcal B_S].\]
Similarly, let $\lambda_T,\lambda_H,\lambda_M$ denote the conditional
distributions given $\mathcal B_T,\mathcal B_H,\mathcal B_M$, respectively.
Writing
\[
\eta=\Pr[\mathbf B\in\mathcal B_S],
\]
and
\[q_T:=\Pr[\mathbf B\in\mathcal B_T\mid \mathbf B\notin\mathcal B_S],\qquad
q_H:=\Pr[\mathbf B\in\mathcal B_H\mid \mathbf B\notin\mathcal B_S],
\]
\[
q_M:=\Pr[\mathbf B\in\mathcal B_M\mid \mathbf B\notin\mathcal B_S],
\]
we have \[q_T+q_H+q_M=1\]
and
\[\lambda_0=q_T\lambda_T+q_H\lambda_H+q_M\lambda_M.
\]
Equivalently,
\[
q_T=\frac{p_T}{1-\eta},
\qquad
q_H=\frac{p_H}{1-\eta},
\qquad
q_M=\frac{p_M}{1-\eta}.
\]

Now, we can state the following coupling $\pi_{1,0}$ between $\mu_1,\mu_0$:
\begin{enumerate}
    \item Sample
    \[
    B_1=(b_1,\ldots,b_\ms)\sim\blockdist_1.
    \]
    Since $B_1\in\mathcal B_S$ and the blocks are disjoint, there is a unique $r\in[\ms]$ such that
    \[
    \{x_1,x_2\}\subseteq b_r.
    \]

    \item Independently choose a case
    \[
    C\in\{T,H,M\}
    \]
    with probabilities
    \[
    \Pr[C=T]=q_T,\qquad
    \Pr[C=H]=q_H,\qquad
    \Pr[C=M]=q_M.
    \]

    \item Conditional on $C$, construct $B_2$ from $B_1$ as follows.

    \begin{itemize}
        \item \textbf{Case $T$:}
        Choose $i\in\{1,2\},\,j\in[\ms]\setminus\{r\},\,z\in b_j$ uniformly at random, and swap $x_i$ and $z$ between $b_r$ and $b_j$.
        The resulting configuration satisfies $B_2\in\mathcal B_T.$

        \item \textbf{Case $H$:} Choose $i\in\{1,2\}$ uniformly at random and choose $z\in [\ns]\setminus\bigcup_{t=1}^{\ms} b_t$
        uniformly at random. Replace $x_i$ by $z$ in $b_r$. The resulting configuration satisfies $B_2\in\mathcal B_H.$

        \item \textbf{Case $M$:}
        Choose an ordered pair of distinct vertices $(z_1,z_2)$,
        with $z_1,z_2\in [\ns]\setminus\bigcup_{t=1}^{\ms} b_t$,
        uniformly at random, and for $i\in\{1,2\}$ replace
        $x_i$ in $b_r$ by $z_i$.
        The resulting configuration satisfies $B_2\in\mathcal B_M.$
    \end{itemize}

    \item Output $(y,y') = (g(B_1),g(B_2))$.
\end{enumerate}
By definition, $y\sim \ome_1 = \blockdist_1 \kappa_g$. We argue that (I) $y'\sim \ome_0 = \blockdist_0\kappa_g$, and (II) for all pairs $(y,y')$ produced by this coupling, $d(y,y')\leq2(h-1)$. First, notice that to show (I) $y'\sim \ome_0$ it is sufficient to show that $B_2\sim\blockdist_0$, since by definition $\ome_0=\blockdist_0\kappa_g$. We first proceed in proving that $B_2\sim \blockdist_0$. Along the way, we will also gain some intuition for the distance bound.

Let $N\defeq |\mathcal{B}|$ be the total number of different configurations, and $N_S\defeq |\mathcal{B}_S|=\eta N$ be the number of configurations that satisfy $S$. Similarly, define $N_T,\,N_H$ and $N_M$. Let $N' = N_T + N_H + N_M=(1-\eta)N$ be the number of configurations that satisfy $T$, $H$ or $M$. Given $B_1\sim\blockdist_1$, which is uniform on $\mathcal{B}_S$, the algorithm described above produces a uniformly random element of $\mathcal{B}\setminus\mathcal{B}_S$. In other words, $B_2\sim\blockdist_0$. To prove this, we fix $X\in\mathcal{B}\setminus\mathcal{B}_S$ and compute $\Pr[B_2=X]$, breaking it up into the three cases $X\in\mathcal{B}_T,\,\mathcal{B}_H$, and $\mathcal{B}_M$. We will show that in each case it equals $1/N'$.

\paragraph{Case $T$.} Recall that in case $T$, both $x_1$ and $x_2$ are sampled from $[\ns]$, but they appear in different blocks, say, $b_1$ and $b_2$. Fix one such $X\in \mathcal{B}_T$ and notice that there are $2(h-1)$ ways to turn it into an instance in $\mathcal{B}_S$ using one swap. We can choose to bring $x_1$ to $x_2$ or vice-versa, and we can choose which non-endpoint in $b_1$ or $b_2$ to swap with. By symmetry, these are also all of the possible choices of $B_1$ that can be made into an instance of $X$ by our coupling. Since $B_1$ is uniform over $\mathcal{B}_S$, the probability that we sample one of these configurations is $\frac{2(h-1)}{N_S}$. In the instance that we sample one of these configurations, we also need to consider the probability that we both enter the $T$ branch and make the unique choice of $(i,j,z)$ that yields $X$. Putting these probabilities together, we see that 
\begin{align*}
    \Pr[B_2 = X] & = q_T  \cdot \frac{2(h-1)}{N_S} \cdot \frac{1}{2h(\ms-1)}\\
        & = \frac{h^2\ms(\ms-1)}{(1-\eta)n(n-1)} \cdot \frac{2(h-1)}{2N_S (\ms-1) h }\\
        & = \frac{\eta}{1-\eta}\cdot \frac{1}{\eta N} = \frac{1}{N'}
\end{align*}

\paragraph{Case $H$.} Recall that the event $H$ corresponds to sampling exactly one of $x_1$ or $x_2$. Fix some choice of $X\in \mathcal{B}_H$, and similarly to before there are $(h-1)$ ways to transform this into an instance of $\mathcal{B}_S$ with one replacement. Just the same as before, this implies that there are $h-1$ choices of $B_1$ that can be made into an instance of $X$. Let $d=\ns - \ms h$ be the number of remaining elements after sampling. Multiplying through the probabilities that we sample a suitable element in $B_1$, that we enter the right path of the algorithm, and that we pick the matching values of $(i,z)$, we see
\begin{align*}
    \Pr[B_2 = X] &= q_H \cdot \frac{h-1}{N_S} \cdot \frac{1}{2\maxdeg}\\
        &= \frac{2\maxdeg\ms h}{(1-\eta)n(n-1)} \cdot \frac{h-1}{N_S 2d}\\
        &= \frac{1}{N'}.
\end{align*}

\paragraph{Case $M$.} Again, fix some choice of $X\in \mathcal{B}_M$ and notice that there are $\ms h(h-1)$ ways to turn this into an instance of $\mathcal{B}_S$: choose a block, and the ordered pair of positions to be occupied by $(x_1,x_2)$. Naturally, for each such choice of $B_1$, there is exactly one ordered pair $(z_1,z_2)$ whose selection in the coupling yields $X$. As such,
\begin{align*}
    \Pr[B_2 = X] & = q_M \cdot \frac{\ms h(h-1)}{N_S} \cdot \frac{1}{d(d-1)}\\
        & = \frac{d(d-1)}{(1-\eta)n(n-1)}\cdot \frac{\ms h(h-1)}{N_S d(d-1)}\\
        & = \frac{1}{N'}
\end{align*}

{
By definition, applying $g$ to the configurations produced in this coupling yields a coupling between $\omega_1$ and $\omega_0$, since those distributions are exactly $\lambda_1 \kappa_g$ and $\lambda_0\kappa_g$. To apply~\cref{def:dcc}, we need to bound the \emph{edge distance} between $y=g(B_1)$ and $y'=g(B_2)$.
}
In case $H$ we replace one vertex with another, changing at most $h-1$ edges. In case $T$ we modify one vertex in $b_r$ and another in $b_j$, this changes up to $2(h-1)$ edges. In case $M$ we change at most $(h-1) + (h-2)=2h - 3$ edges. For $\ms\geq 2$ the largest of these is case $T$, so we will use $2(h-1)$ as our group-privacy distance for the remainder of this proof.\footnote{In the case that $\ms=1$, there is no other $b_j$ to be affected in a move, meaning the right distance for group privacy in this case is $2h-3$.}

Similarly to~\cref{eq:distance-bound-11}, applying~\cref{def:dcc} with $\pi_{1,0}$, where any pair of graph tuples $(y,y')$ have $d(y,y')\leq 2(h-1)$, we get that for every $\priv>0$,
\begin{equation}
    D_\eep(\mu_1\|\mu_0) \leq \sum_{y,y'}\pi_{1,0}(y,y')D_\eep (\pmech(y)\|\pmech(y')) \leq \pdelta_{\pmech,2(h-1)}(\priv).\label{eq:distance-bound-10}
\end{equation}
Since $\pmech$ is $(\priv_0,\pdelta_0)$-edge differentially private, group privacy gives
\begin{equation*}
    \pdelta_{\pmech,2(h-1)}(2(h-1)\priv_0) \leq \frac{e^{2(h-1)\priv_0} - 1}{\eepz - 1}\pdelta_0\quad\text{and}\quad\pdelta_\pmech(2(h-1)\priv_0) \leq \pdelta_\pmech(\priv_0) \leq \pdelta_0,
\end{equation*}
so long as $h\geq 2$. The second term uses the fact that privacy profiles are non-increasing in $\priv$. 

We now combine these bounds using advanced joint convexity. Set\[
\diverge\defeq e^{2(h-1)\priv_0},\qquad \diverge'\defeq 1 + \eta(\diverge - 1) = \eepp,\qquad \varphi\defeq \diverge'/\diverge,
\]
so that \[
\priv' = \log( 1 + \eta (e^{2(h-1)\priv_0} - 1)).
\]
By~\cref{advanced-joint-convexity}, using the decompositions $\mu=(1-\eta)\mu_0 + \eta \mu_1$ and $\mu' = (1-\eta) \mu_0 + \eta \mu_1'$, followed by joint convexity and substituting in~\cref{eq:distance-bound-11,eq:distance-bound-10} with $\priv=2(h-1)\priv_0$,
\[
D_\eepp (\mu\| \mu') \leq \eta\left[ (1-\varphi) \frac{e^{2(h-1)\priv_0} - 1}{\eepz -1}\pdelta_0 + \varphi\pdelta_0 \right] \leq \eta \frac{e^{2(h-1)\priv_0} - 1}{\eepz -1}\pdelta_0
\]
where the final inequality holds because $\pdelta_0 \leq \frac{e^{2(h-1)\priv_0}- 1}{\eepz -1}\pdelta_0$ for $h\geq 2$. 

An identical argument, exchanging the roles of $G$ and $G'$, shows the same bound for $D_{\eepp}(\mu'\|\mu)$. Hence, by~\cref{thm:dp-profile}, the subsampled mechanism satisfies $(\priv',\pdelta')$-edge differential privacy with
\[
\priv' = \log\left( 1 + \frac{\ms h(h-1)}{\ns(\ns-1)}\left( e^{2(h-1)\priv_0} - 1 \right) \right),\qquad \pdelta' = \frac{\ms h(h-1)}{\ns(\ns-1)}\frac{e^{2(h-1)\priv_0 } - 1}{\eepz - 1}\pdelta_0.\qedhere
\]

\end{proofof}

\section{Various Non-private Testing Algorithms}
\label{sec:algorithms}

In this section, we present the non-private algorithms from the property testing literature that we employ to design their private counterparts. We will present the algorithms for $H$-freeness testing and Bipartiteness testing in the bounded-degree model.

\subsection{\texorpdfstring{$H$}{H}-freeness testing algorithm in the dense graph model}

We now present the $H$-freeness tester in the dense graph model.

\begin{algorithm}[H]
\caption{$H$-freeness tester in dense graphs}\label{alg:Hfreenessdensegraph}

\KwIn{Query access to a dense graph $G=([\ns],E)$, a known graph $H$, proximity parameter $\dst$, failure probability $\failp$.}
\KwOut{\texttt{Accept} or \texttt{Reject}.}

\For{$\frac{\log(1/\failp)}{\prox(\dst,H)}$ times}{

Sample a set $S$ of $|V(H)|$ vertices of $G$ uniformly at random.

\If{the induced subgraph $G[S]$ contains $H$ as a subgraph}{

Return \texttt{Reject}

}

}

Return \texttt{Accept}.

\end{algorithm}

\subsection{\texorpdfstring{$H$}{H}-freeness testing algorithm in the bounded-degree model}

This algorithm is divided into two subroutines. In \cref{alg:pot-h-free}, we present a proximity-oblivious algorithm, based on which in \cref{alg:gen-h-free} we design the final algorithm for $H$-freeness testing. This algorithm is from \cite{goldreich2002property}.

\begin{algorithm}[H]
\caption{Proximity Oblivious $H$-freeness tester}\label{alg:pot-h-free}

\KwIn{Query access to a graph $G=([\ns],E)$ with maximum degree $\maxdeg$, $H$ with radius $\rd$.}
\KwOut{\texttt{Accept} or \texttt{reject}.}

Uniformly select a vertex $v\in[\ns]$.

Conduct a BFS of depth at most $\rd$.

\texttt{Accept} if and only if the explored subgraph is $H$-free.
\end{algorithm}

\begin{algorithm}[H]
\caption{Testing $H$-freeness in bounded-degree model}\label{alg:gen-h-free}

\KwIn{Query access to a graph $G=([\ns],E)$ with maximum degree $\maxdeg$, $H$ with radius $\rd$, proximity parameter $\dst$, failure probability $\failp$.}
\KwOut{\texttt{Accept} or \texttt{Reject}.}

\For{$i=1$ to $\bigO{\frac{\log(1/\failp)}{\dst}}$}{

Call \Cref{alg:pot-h-free} (the proximity oblivious tester).

}

\texttt{Accept} if and only if all of the 
calls to \Cref{alg:pot-h-free} have returned \texttt{Accept}.

\end{algorithm}

\subsection{Bipartiteness testing algorithm}

Now we present the non-private bipartiteness testing algorithm from \cite{goldreich1999a}.

\begin{algorithm}[H]
\caption{Bipartiteness testing in bounded-degree graphs}\label{alg:biptest}

\KwIn{Query access to a graph $G=([\ns],E)$ with maximum degree $\maxdeg$, proximity parameter $\dst$, failure probability $\failp$.}
\KwOut{\texttt{Accept} or \texttt{Reject}.}

\For{$O(\frac{1}{\dst})$ times}{

Sample a vertex $v \in [n]$ uniformly at random.

Perform $O(\sqrt{n} \mathsf{poly}(\log n \dst))$ random walks from $v$, each with length $O(\mathsf{poly}(\log n, \dst))$.

Let $R_0$ and $R_1$ be the set of vertices reached from $v$ by even (odd) length random walks.

\If{$R_0 \cap R_1 \neq \emptyset$}{

Return \texttt{Reject}.

}

}

Return \texttt{Accept}.

\end{algorithm}

\begin{algorithm}
\caption{Tester for hyperfinite properties}\label{alg:hyperfinite-tester}

\KwIn{Query access to a graph $G=([\ns],E)$ with maximum degree $\maxdeg$, $\dst$, $\failp$, $\rho$, \property.}
\KwOut{\texttt{Accept} or \texttt{reject}.}

\tcp{Phase 1 Constants (Test that $G$ is $\rho$-hyperfinite)}
Let $\gamma_0 = \dst/4$, and $\dst_1 = \gamma_0^3 / 54000$\;
Let $\gamma_1 = \dst_1 / (4\maxdeg)$, and $q = \lceil \gamma_1^{-2} \ln(4/\failp) \rceil$\;
Let $k_1 = \rho(\gamma_1^3 / 54000)$\;

\tcp{Phase 2 Constants (Test for the property \property)}
Let $\eta = \eta(\dst, \rho, \maxdeg)$ and $D = D(\dst, \rho, \maxdeg)$ be the parameters guaranteed by Theorem~3.1 in \cite{newman2013every}\;
Let $N = N(D, \maxdeg)$ be the number of non-isomorphic rooted graphs of radius $\le D$\;
Let $s = \lceil \eta^{-2} N^2 \ln(20N/\failp) \rceil$\;
\BlankLine
\tcp{Phase 1: Test that G is $\rho$-hyperfinite}
Initialize marked queries count $M = 0$\;
\For{$i = 1$ \KwTo $q$}{
    Select a vertex $v_i \in [\ns]$ uniformly at random\;
    Select an integer $r_i \in [\maxdeg]$ uniformly at random\;
    \If{$\deg(v_i) \ge r_i$}{
        Let $u_i$ be the $r_i$-th neighbor of $v_i$\;
        Run the local partition oracle (as stated in Lemma~4.3 in \cite{newman2013every}) with parameter $\gamma_1$ on $v_i$\;
        Let $P[v_i]$ be the component (of size $\le k_1$) returned by the oracle\;
        \If{$u_i \notin P[v_i]$}{
            $M \gets M + 1$ \tcp*{The edge $(v_i, u_i)$ is cut}
        }
    }
}
\If{$M > 4\gamma_1 q$}{
    \Return{\texttt{reject}}\;
}
\BlankLine
\tcp{Phase 2: Test for the property \property}
Initialize an empirical frequency vector $\hat{f}$ of size $N$ with zeros\;
\For{$j = 1$ \KwTo $s$}{
    Select a vertex $w_j \in [\ns]$ uniformly at random\;
    Let $\ell \in [N]$ be the index of the isomorphism type of the raw $D$-disc of $w_j$\;
    $\hat{f}[\ell] \gets \hat{f}[\ell] + 1/s$\;
}
\BlankLine
\texttt{Accept} if, and only if, there exists a valid graph $G' \in \property$ such that its exact raw $D$-disc frequency vector $f^*$ satisfies $\| \hat{f} - f^* \|_1 \le \eta$. Otherwise, \texttt{reject}.
\end{algorithm}

\section{Impossibility of Pure DP in one-sided error models}
\label{sec:pure-dp}

In this section, we prove that one-sided error pure-DP tester is impossible regardless of the graph model.

\begin{theorem}
\label{prop:pure-dp}
Let $\property$ be any non-empty graph property (contains at least one $n$-vertex
graph with that property).  Any one-sided error tester $\tester$ for~$\property$ that is
$(\priv, 0)$-edge-DP (pure edge DP) must satisfy
$\Pr[\tester(G) = \mathrm{reject}] = 0$ for every graph $G$.
\end{theorem}

\begin{proof}
Consider a graph $G_0 \in \property$. Since $\tester$ is one-sided error, this implies that $\Pr[\tester(G_0) = \mathrm{reject}] = 0$. For any graph~$G^*$, note that there is a path of edge-neighboring graphs $G_0 \sim_e G_1 \sim_e \cdots \sim_e G_m = G^*$ with
$m \leq \binom{n}{2}$ (this bound holds for both dense and bounded degree models).  By applying pure DP inductively, we have the following:
\[
  \Pr[\tester(G_{i+1}) = \mathrm{reject}]
    \;\leq\;
    e^{\priv} \cdot \Pr[\tester(G_i) = \mathrm{reject}]
    \;=\;
    e^{\priv} \cdot 0
    \;=\; 0.
\]
Hence $\Pr[\tester(G) = \mathrm{reject}] = 0$ for all graphs $G$. This completes the proof.    
\end{proof}

While it may be possible to bypass this bottleneck and design \emph{approximate}-DP testers with one-sided error, in our work, we choose to focus on the strongest privacy guarantees and thus, design pure DP testers with two-sided error.

\end{document}